\documentclass[letterpaper,12pt,twoside]{article}

\RequirePackage[OT1]{fontenc}
\RequirePackage{amsthm,amsmath,natbib}
\RequirePackage[colorlinks,citecolor=blue,urlcolor=blue]{hyperref}

\usepackage{amssymb}
\usepackage[utf8]{inputenc}
\usepackage{bbm}
\usepackage{xcolor}
\usepackage{graphicx}
\usepackage{multirow}
\usepackage{lscape}
\usepackage[figuresright]{rotating}

\def\bSig\mathbf{\Sigma}

\newcommand{\rnc}{\renewcommand}
\newcommand{\nc}{\newcommand}
\newcommand{\mrm}{\mathrm}
\renewcommand{\b}{\textbf}
\nc{\mb}{\mathbb}
\nc{\mc}{\mathcal}
\nc{\E}{\mb{E}}
\nc{\N}{\mb{N}}
\nc{\R}{\mb{R}}
\nc{\Q}{\mb{Q}}
\rnc{\P}{\mrm P}
\rnc{\d}{\mrm d}
\nc{\C}{\mc{C}}
\nc{\D}{\mc{D}}
\nc{\B}{\mc{B}}
\nc{\oPo}{\stackrel{\mrm P}{\rightarrow}}
\nc{\oWo}{\stackrel{w}{\rightarrow}}
\nc{\oDo}{\stackrel{d}{\longrightarrow}}
\rnc{\b}{\textbf}
\nc{\bs}{\boldsymbol}

\numberwithin{equation}{section}
\theoremstyle{plain}
\newtheorem{theorem}{Theorem}[section]
\newtheorem{corollary}{Corollary}[section]
\theoremstyle{remark}
\newtheorem{remark}{Remark}[section]
\newtheorem{example}{Example}[section]
\newtheorem{lemma}{Lemma}[section]

\begin{document}

\title{The Wild Bootstrap for Multivariate Nelson-Aalen Estimators}

\author{\hspace{-1.2cm}\footnote{Authors with equal contribution, e-mail: tobias.bluhmki@uni-ulm.de, dennis.dobler@uni-ulm.de, jan.beyersmann@uni-ulm.de, markus.pauly@uni-ulm.de} \ Tobias Bluhmki, 
${}^*$ Dennis Dobler,
Jan Beyersmann,
Markus Pauly \\
Ulm University, Institute of Statistics, \\ Helmholtzstrasse 20, 89081 Ulm, \\ Germany}

\maketitle



\begin{abstract}
We rigorously extend the widely used wild bootstrap resampling technique to the multivariate Nelson-Aalen estimator under Aalen's multiplicative intensity model. Aalen's model covers general Markovian multistate models including competing risks subject to independent left-truncation and right-censoring. 
{\color{black}
This leads to various statistical applications such as asymptotically valid confidence bands or tests for equivalence and proportional hazards.
This is exemplified in a data analysis examining the impact of ventilation on the duration of intensive care unit stay. 
The finite sample properties of the new procedures are investigated in a simulation study.}
\end{abstract}


\section{Introduction}

One of the most crucial quantities within the analysis of time-to-event data with independently right-censored and left-truncated survival times is the cumulative hazard function, also known as transition intensity. Most commonly, it is nonparametrically estimated by the well-known \emph{Nelson-Aalen estimator} \citep[Chapter IV]{abgk93}. In this context, time-simultaneous confidence bands are the perhaps best interpretative tool to account for related estimation uncertainties. 

The construction of confidence bands is typically based on the asymptotic behavior of the underlying stochastic processes, more precisely, the (properly standardized) Nelson-Aalen estimator asymptotically behaves like a Wiener process. Early approaches utilized this property to derive confidence bands for the cumulative hazard function; see e.g., \citet{bie87} or Section~IV.1.3 in \citet{abgk93}.

However, \citet{dudek08} found that this approach applied to small samples can result in considerable deviations from the aimed nominal level. To improve small sample properties, \citet{efron79,efron81} suggested a computationally convenient and flexible resampling technique, called \textit{bootstrap}, where the unknown non-Gaussian quantile is approximated via repeated generation of point estimates based on random samples of the original data. For a detailed discussion within the standard right-censored survival setup, see also \citet{akritas86}, \citet{lo86}, and \citet{horvath87}. The simulation study of \citet{dudek08} particularly reports improvements of bootstrap-based confidence bands for the hazard function as compared to those using asymptotic quantiles. An alternative is the so-called \textit{wild bootstrap} firstly proposed in the context of regression analyses \citep{wu86}. As done in \citet{lin93}, the basic idea is to replace the (standardized) residuals with independent standardized variates -- so-called multipliers -- while keeping the data fixed. One advantage compared to Efron's bootstrap is to gain robustness against variance heteroscedasticity \citep{wu86}. Using standard normal {\color{black}multipliers}, this resampling procedure has been applied to construct time-simultaneous confidence bands for survival curves under the Cox proportional hazards model \citep{lin1994} and adapted to cumulative incidence functions in the more general competing risks setting \citep{lin97}. The latter approach has recently been extended to general wild bootstrap multipliers with mean zero and variance one \citep{beyersmann12b}, which indicate possible improved small sample performances. This result was confirmed in \citet{dobler14} as well as \citet{dobler15a}, where more general resampling schemes are discussed.

In contrast to probability estimation, the present article focuses on the nonparametric estimation of  cumulative hazard functions and proposes a general and flexible wild bootstrap resampling technique, which is valid for a large class of time-to-event models. In particular, the procedure is not limited to the standard survival or competing risks framework. The key assumption is that the involved counting processes satisfy the so-called \textit{multiplicative intensity model} \citep{abgk93}. Consequently, arbitrary Markovian multistate models with finite state space are covered, as well as various other intensity models \citep[e.g., excess or relative mortality models, cf.][]{andersen1989} and specific semi-Markov situations \citep[][Example X.1.7]{abgk93}. Independent right-censoring and left-truncation can straightforwardly be incorporated.

The main aim of this article is to mathematically justify the wild bootstrap technique for the multivariate Nelson-Aalen estimator in this general framework. This is accomplished by generalizing core arguments in \citet{beyersmann12b} and \citet{dobler14} and verifying conditional tightness via a modified version of Theorem~15.6 in \citet{billingsley68}; see p. 356 in \citet{jacod03}. Compared to the standard survival or competing risks setting, with at most one transition per individual, the major difficulty is to account for counting processes having an arbitrarily large random number of jumps. As \citet{beyersmann12b} suggested in the competing risks setting, we also permit for more general multipliers with expectation 0 and variance 1 and extend the resulting weak convergence theorems to resample the multivariate Nelson-Aalen estimator in our general setting. For practical applications, this result allows, for instance, within- or two-sample comparisons and the formulation of statistical tests. 

The wild bootstrap is exemplified to statistically assess the impact of mechanical ventilation in the intensive care unit (ICU) on the length of stay. A related problem is to investigate ventilation-free days, which was established as an efficacy measure in patients subject to acute respiratory failure \citep{schoenfeld2002}. However, statistical evaluation of their often-used methodology (see e.g., \citealp{sauaia2009} and \citealp{stewart2009}) relies on the constant hazards assumption. Other publications like \citet{dewit2008}, \citet{trof2012}, or \citet{curley2015} used a Kaplan-Meier-type procedure that does not account for the more complex multistate structure. In contrast, we propose an illness-death model with recovery that methodologically works under the more general time-inhomogeneous Markov assumption and captures both the time-dependent structure of mechanical ventilation and the competing endpoint `death in ICU'.

The remainder of this article is organized as follows: Section~\ref{sec:model} introduces cumulative hazard functions
and their Nelson-Aalen-type estimators using counting process formulations.
After summarizing its asymptotic properties, 
Section~\ref{sec:main} offers our main theorem on conditional weak convergence for the wild bootstrap. 
This allows for various statistical applications in Section~\ref{sec:stat_app}:
Two-sided hypothesis tests and various sorts of time-simultaneous confidence bands are deduced,
as well as simultaneous confidence intervals for a finite set of time points. 
{\color{black}
Furthermore, tests for equivalence, inferiority and superiority as well as for proportionality of two hazard functions constitute useful criteria in practical data analyses.
}
A simulation study assessing small and large sample performances of both the derived confidence bands in comparison to the algebraic approach based on the time-transformed Brownian motion {\color{black}and the  tests for proportional hazards} is reported in Section~\ref{sec:simus}. The SIR-3 data on patients in ICU (\citealp{beyersmann2006} and \citealp{wolkewitz2008}) serves as its template and is practically revisited in Section~\ref{sec:dataex}. 
Concluding remarks and a discussion are given in Section~\ref{sec:discussion}. 
All proofs are deferred to the Appendix. 

\section{Non-Parametric Estimation under the Multiplicative Intensity Structure}
\label{sec:model}

Throughout, we adopt the notation of \cite{abgk93}.
For ${k\in \mathbb{N}}$, let $\textbf{N}=\left(N_1,\ldots,N_k\right)'$ be a multivariate counting process
which is adapted to a filtration $(\mathcal F_t)_{t \geq 0}$. 
Each entry $N_j, j=1,\dots,k,$ is supposed to be a c\`adl\`ag function, zero at time zero, and to have piecewise constant paths with jumps of size one. 
In addition, assume that no two components jump at the same time and that each $N_{j}(t)$ satisfies the \textit{multiplicative intensity model} of \cite{aalen1978} 
with intensity process given by $\lambda_j(t) = \alpha_j(t) Y_j(t)$.
Here, $Y_j(t)$ defines a predictable process not depending on unknown parameters
and $\alpha_j$ describes a non-negative (hazard) function.
For well-definiteness, the observation of $\textbf{N}$ is restricted to the interval $[0,\tau]$, where 
$\tau< \tau_j = {\sup}\big\{ u \geq 0 :\int_{(0,u]}\alpha_j(s)ds<\infty\big\} \ \text{for all } j = 1, \dots, k.$ 
The multiplicative intensity structure covers several customary frameworks in the context of time-to-event analysis.
The following overview specifies frequently used models.
\begin{example}\label{ex:cov_models}
(a) Markovian multistate models with finite state space $\mathcal S$ are very popular in biostatistics. 
In this setting, $Y_\ell(t)$ represents the total number of individuals in state $\ell$ just prior to $t$ (`number at risk'), 
whereas $\alpha_{\ell m}(t)$ is the instantaneous risk (`transition intensity') to switch from state $\ell$ to $m$, 
where $\ell, m \in \mathcal S$, $\ell\ne m$.
Here $N_\ell = \sum_{i=1}^n N_{\ell;i}$ is the aggregation over individual-specific counting processes with $n \in \mathbb N$ individuals under study.
For specific examples (such as competing risks or the illness-death model) and details including the incorporation of independent left-truncation and right-censoring, see \cite{abgk93} and \cite{aalen08}. \\
(b) Other examples are the relative or excess mortality model, where not all individuals necessarily share the same hazard rate $\alpha$. 
In this case $Y$ cannot be interpreted as the total number of individuals at risk as in part (a); 
see Example IV.1.11 in \cite{abgk93} for details. \\
(c) The time-inhomogeneous Markov assumption required in part (a) can even be relaxed in specific situations: Following Example X.1.7 in \cite{abgk93}, consider an illness-death model without recovery. {\color{black} Assuming that the transition intensity $\alpha_{12}$ depends on the duration $d$ in the intermediate state, but not on time $t$, leads to semi-Markov process not satisfying the multiplicative intensity structure. This is because the intensity process of $N_{12}(t)$ is given by $\alpha_{12}(t-T)Y_1(t)$, where the first factor of the product is not deterministic anymore. Here, $T$ is the random transition time into state 1. However, when $d=t-T$ is used as the basic timescale, the counting process $K(d)=N_{12}(d+T)$ has intensity $\alpha_{12}(d)Y_1(d)$ with respect to the filtration 
\begin{align*}
\mathcal F_d=\left(\sigma\lbrace (N_{01}(t), N_{02}(t)): 0<t<\tau\rbrace \lor \sigma\lbrace K(d): 0<d<\infty\rbrace \right).
\end{align*}
Thus, the multiplicative intensity structure is fulfilled. }
\end{example}
Under the above assumptions, the Doob-Meyer decomposition applied to $N_j$
leads to
\begin{eqnarray}
 dN_{j}(s)= \lambda_{j}(s)ds+dM_{j}(s),\label{eq:doobmeyer}
\end{eqnarray}
where the $M_{j}$ are zero-mean martingales with respect to $(\mathcal F_t)_{t \in [0,\tau]}$. The canonical nonparametric estimator of the cumulative hazard function
$A_j(t)=\int_{(0,t]}\alpha_j(s)ds$ is given by the so-called \textit{Nelson-Aalen estimator} 
\begin{eqnarray*}
\hat A_{jn}(t)=\int\limits_{(0,t]}\frac{J_j(s)}{Y_j(s)}dN_j(s). \label{eq:NA}
\end{eqnarray*}
Here, $J_j(t)=\mathbf 1\{Y_j(t)>0\}$,
$\frac00 := 0$, and $n \in \mathbb N$ is a sample size-related number (that goes to infinity in asymptotic considerations). Its multivariate counterpart is introduced by $\hat{\textbf{{A}}}_n:=(\hat A_{1n},\ldots,\hat A_{kn})'$.
As in \cite{abgk93}, suppose that there exist deterministic functions $y_j$ 
with $\inf_{u \in [0,\tau]} y_j(u)>0$ such that 
\begin{eqnarray} \underset{s\in[0,\tau]}{\sup}\left|\frac{Y_j(s)}{n}-y_j(s)\right|\xrightarrow[]{\text{ }P\text{ }}0 \quad \text{for all } j=1,\dots, k , \label{Pre.ass11}
\end{eqnarray}
where `$\stackrel{P}{\rightarrow}$' denotes convergence in probability for $n\rightarrow\infty$.
For each $j$, define the normalized Nelson-Aalen process $W_{jn}:=\sqrt{n} ( \hat A_{jn}-A_j )$ 
possessing the asymptotic martingale representation 
\begin{eqnarray}
W_{jn}(t)& \doteqdot &\sqrt{n}\int\limits_{(0,t]}\frac{J_j(s)}{Y_j(s)}dM_j(s) \label{NA_uni_martingale}
\end{eqnarray} 
with $M_{j}$ given by (\ref{eq:doobmeyer}). 
Here, `$\doteqdot$' means that the difference of both sides converges to zero in probability.
Define the vectorial aggregation of all $W_{jn}$ as
$\textbf{\textit{W}}_n = (W_{1n}, \dots, W_{kn})'$ and let `$\stackrel{d}{\rightarrow}$' denote convergence in distribution for $n\rightarrow\infty$.
Then, Theorem IV.1.2 in \cite{abgk93} in combination with (\ref{Pre.ass11}) provides a weak convergence result 
on the $k$-dimensional space $\mathfrak D[0,\tau]^k$ of c\`{a}dl\`{a}g functions endowed with the product Skorohod topology. 
\begin{theorem}\label{Th.Na.uni}
If assumption (\ref{Pre.ass11}) holds, we have convergence in distribution 
\begin{eqnarray}
\textbf{W}_{n} \stackrel{d}{\longrightarrow} \textbf{U} = (U_1, \dots, U_k)',
\end{eqnarray}
on $\mathfrak D[0,\tau]^k$, where $U_1, \dots, U_k$ are independent zero-mean Gaussian martingales with covariance functions 
$\psi_j(s_1,s_2):=Cov(U_j(s_1),U_j(s_2))={\int}_{(0,s_1]}\frac{\alpha_j(s)}{y_j(s)}ds$ for $j = 1, \dots, k$ and $0\le s_1\le s_2\le \tau$.
\end{theorem}
The covariance function $\psi_j$ is commonly approximated by the \textit{Aalen-type}
 \begin{eqnarray} \hat\sigma_j^2(s_1) = n \int\limits_{(0,s_1]} \frac{J_j(s)}{Y_j^2(s)} d N_j(s). \label{eq:varaalen}\end{eqnarray}
 or the \textit{Greenwood-type} estimator
 \begin{eqnarray} \hat\sigma_j^2(s_1) = n \int\limits_{(0,s_1]} \frac{J_j(s) (Y_j(s)-\Delta N_j(s))}{Y_j^3(s)} d N_j(s) \label{eq:vargreenwood}\end{eqnarray}
 which are consistent for $\psi_j(s_1,s_2)$ under the assumption of Theorem~\ref{Th.Na.uni}; cf. (4.1.6) and (4.1.7) in \cite{abgk93}. Here, $\Delta N_j(s)$ denotes the jump size of $N_j$ at time $s$.

\section{Inference via Brownian Bridges and the Wild Bootstrap}
\label{sec:main}

As discussed in \citet{abgk93}, the limit process $\textbf{\textit{U}}$ can analytically be approximated via Brownian bridges. However, improved coverage probabilities in the simulation study in Section \ref{sec:simus} suggest that the proposed wild bootstrap approach may be preferable. 
First, we sum up the classic result.

\subsection{Inference via Transformed Brownian Bridges}
\label{sec:brownian_bridge}
 The asymptotic mutual independence stated in Theorem~\ref{Th.Na.uni} allows to focus on a single component of $\textbf{\textit{W}}_n$, 
 say $W_{1n} = \sqrt{n} ( \hat A_{1n} - A_1 )$. For notational convenience, we suppress the subscript $1$. 
 Let $g$ be a positive (weight) function on an interval $[t_1,t_2]\subset [0,\tau]$ of interest and $B^0$ a standard Brownian bridge process. Then, as $n \rightarrow \infty$, it is established in Section~IV.1 in \cite{abgk93} that
 \begin{eqnarray}
  \sup_{s \in [t_1,t_2]} \Big| \frac{\sqrt n ( \hat A_n(s) - A(s))}{1 + \hat\sigma^2(s)} g \Big(\frac{\hat\sigma^2(s)}{1 + \hat\sigma^2(s)} \Big) \Big|
    \stackrel{d}{\longrightarrow} \sup_{s \in [\phi(t_1),\phi(t_2)]} | g(s) B^0(s) |. \label{eq:asymptdist}
 \end{eqnarray}
 Here $\phi(t) = \frac{\sigma^2(t)}{1 + \sigma^2(t)}$, $\sigma^2(t) = \psi(t,t)$ 
 and $\hat\sigma^2(t)$ is a consistent estimator for $\sigma^2(t)$, such as~\eqref{eq:varaalen} or~\eqref{eq:vargreenwood}.
 Quantiles of the right-hand side of \eqref{eq:asymptdist} for $g\equiv 1$ are recorded in tables \citep[e.g.,][]{koziol1975,hall1980,schumacher1984}. For general $g$, they can be approximated via standard statistical software.

 Even though relation \eqref{eq:asymptdist} enables statistical inference based on the asymptotics of a central limit theorem, appropriate resampling procedures usually showed improved properties; see e.g., \cite{hall91}, \cite{good05} and \cite{pauly15}.
 
\subsection{Wild Bootstrap Resampling}\label{sec:wbresampling}
In contrast to, for instance, a competing risks model where each counting process $N_{j}$ is at most $n$, 
the number $N_{j}(\tau)$  is not necessarily bounded  in our setup only assuming Aalen's multiplicative intensity model.
Hence, a modification of the multiplier resampling scheme under competing risks suggested by \cite{lin97} and elaborated by \cite{beyersmann12b} is required.  
For this purpose, introduce counting process-specific stochastic processes indexed by $s \in [0,\tau]$ that are independent of $N_j, Y_j$ for all $j=1,\dots,k$.
Let $(G_j(s))_{s \in [0,\tau]}, 1 \leq j \leq k,$ be independently and identically distributed (i.i.d.) white noise processes such that each $G_j(s)$ satisfies $\E(G_j(s)) = 0$ and $var(G_j(s))=1$, $j=1,\dots,k$, $s \in [0,\tau]$.
{\color{black}That is, all $\ell$-dimensional marginals of $G_1$, $\ell \in \N$, shall be the same $\ell$-fold product-measure.}
Then, a \emph{wild bootstrap version} of the normalized multivariate Nelson-Aalen estimator $\textbf{\textit{W}}_n$ is defined as
\begin{eqnarray}
\hat{\textbf{\textit{W}}}_{n}(t) & = & (\hat{W}_{1n}(t), \dots, \hat{W}_{kn}(t) )' \label{NA.uni.Wnh} \\
  & := & \sqrt{n} \bigg( \underset{(0,t]}{\int}\frac{J_1(s)}{Y_1(s)} G_{1}(s) dN_{1}(s),
  \dots, \underset{(0,t]}{\int}\frac{J_k(s)}{Y_k(s)} G_{k}(s) dN_{k}(s) \bigg)'. \nonumber
\end{eqnarray}
In words, $\hat{\textbf{\textit{W}}}_{n}$ is obtained from representation \eqref{NA_uni_martingale} of $\textbf{\textit{W}}_n$
by substituting the unknown individual martingale processes $M_{j}$ 
with the \textit{observable} quantities $G_{j} N_{j}$. 
Even though only the values of each $G_j$ at the jump times of $N_j$ are relevant,
this construction in terms of white noise processes enables a consideration of the wild bootstrap process on a product probability space;
see the Appendix for details.

{\color{black}
Consider for a moment the special case of a multistate model with $n$ i.i.d. individuals (Example~\ref{ex:cov_models}(a)).
For instance, the competing risks model in \cite{lin97} involves at most one transition (and thus one multiplier) per individual,
whereas \cite{glidden02} allows for arbitrarily many transitions but also introduces only one multiplier per individual.
In contrast, our resampling approach is a completely new approach in the sense that it involves independent weightings of all jumps even within the same individual.
Being able to resample the Nelson-Aalen estimator even for randomly many numbers of events per individual in this way is a real novelty
and this problem has not yet been theoretically discussed before -- using any technique whatsoever.
Hence, utilizing white noise processes as done in~\eqref{NA.uni.Wnh} is a new aspect in this area.
}

The limit distribution of $\hat{\textbf{\textit{W}}}_n$ may be approximated by simulating a large number of replicates of the $G$'s, 
while the data is kept fixed. 
For a competing risks setting with standard normally distributed multipliers, our general scheme reduces to the one discussed in \cite{lin97}.

For the remainder of the paper, we summarize the available data 
in the $\sigma$-algebra
${\color{black}\mathcal{C}_0} = \sigma \{N_{j}(u),$ $Y_{j}(u):j = 1,\dots,k,\ u\in[0,\tau]\}.$
{\color{black}A natural way to introduce a filtration based on $\mathcal C_0$ that progressively collects information on the white process is by setting
\begin{align*}
 \mathcal{C}_t = \mathcal C_0 \  \vee \  \sigma\{ G_j(s): j=1,\dots,k, \ s \in [0,t] \}.
\end{align*}
The following lemma is a key argument in an innovative, martingale-based consistency proof of the proposed wild bootstrap technique.
\begin{lemma}
\label{lem:mart}
 For each $n \in \N$, the wild bootstrap version of the multivariate Nelson-Aalen estimator $(\hat{\textbf{\textit{W}}}_{n}(t))_{t \in [0,\tau]}$
 is a square-integrable martingale with respect to the filtration $(\mathcal{F}_t)_{t \in [0,\tau]}$.
 with orthogonal components.
 Its predictable variation process is given by 
 $$ \langle \hat{\textbf{\textit{W}}}_{n} \rangle : \ t \ \longmapsto \ n  \bigg( \int_0^t \frac{J_1(s)}{Y_1^2(s)} d N_1(s), \dots, \int_0^t \frac{J_k(s)}{Y_k^2(s)} d N_k(s) \bigg)$$
 and its optional variation process by 
 $$ [ \hat{\textbf{\textit{W}}}_{n} ] : \ t \ \longmapsto \ n  \bigg( \int_0^t \frac{J_1(s)}{Y_1^2(s)} G_1^2(s) d N_1(s), \dots, \int_0^t \frac{J_k(s)}{Y_k^2(s)} G_k^2(s) d N_k(s) \bigg) .$$
\end{lemma}
}

The following conditional weak convergence result justifies the approximation of the limit distribution of $\textbf{\textit{W}}_{n}$ via $\hat{\textbf{\textit{W}}}_n$ given {\color{black}$\mathcal{C}_0$}. 
Both, the general framework requiring only Aalen's multiplicative intensity structure
as well as using possibly non-normal multipliers are original to the present paper.
\begin{theorem}\label{Th.Na.uni-What}
Let $\textbf U$ be as in Theorem \ref{Th.Na.uni}. 
Assuming \eqref{Pre.ass11}, we have the following conditional convergence in distribution on $\mathfrak D[0,\tau]^k$ given {\color{black}$\mathcal{C}_0$ as $n \rightarrow \infty$}:
\begin{eqnarray*}
\hat{\textbf W}_{n}\stackrel{d}{\longrightarrow} \textbf U \quad \text{in probability.} 
\end{eqnarray*}
\end{theorem}
\begin{remark}\label{rm:wc}
 {\color{black}
 (a) 
 It is due to the martingale property of the wild bootstrapped multivariate Nelson-Aalen estimator
 that we anticipate a good finite sample approximation of the unknown distribution of the Nelson-Aalen estimator.
 In particular, the wild bootstrap, realized by white noise processes as above, succeeds in imitating the martingale structure of the original Nelson-Aalen estimator.
 The predictable variation process of the wild bootstrap process equals the optional variation process of the centered Nelson-Aalen process.
 Hence, both processes share the same properties and approximately the same covariance structure.
 }
 \\ 
 (b)
 Additionally to the proof presented in the Appendix, a more elementary consistency proof is shown in the Supplementary Material.
 The proof transfers the core arguments of \cite{beyersmann12b} and \cite{dobler14} to the multivariate Nelson-Aalen estimator in a more general setting: 
First, we show convergence of all finite-dimensional conditional marginal distributions of $\hat {\textbf{\textit{W}}}_n$ towards $\textbf{\textit{U}}$ generalizing some findings of \cite{Pauly11a}. 
Second, we verify conditional tightness by applying a variant of Theorem~15.6 in \cite{billingsley68}; see \cite{jacod03}, p. 356.
In both cases the subsequence principle for random elements converging in probability is combined with assumption (\ref{Pre.ass11}).
\\
{\color{black}
(c)
  Suppose that $E(n^k J_1(u) / Y_1^k(u)) = O(1)$ for some $k \in \N$ and all $u \in [0,\tau]$,
  which for example holds for any $k \in \N$ if $Y_1$ has a number at risk interpretation.
  Since different increments of $\textbf{\textit{W}}_n$ (to arbitrary powers) are uncorrelated,
  it can be shown that the convergence in Theorem~\ref{Th.Na.uni} for single $t \in [0,\tau]$ even holds in the Mallows metric $d_p$ for any even $0 < p \leq k$;
  see e.g. \cite{bickel81} for such theorems related to the classical bootstrap.
  Provided that the $r$th moment of $G_1(u)$ exists, similar arguments show 
  that the convergence in probability in Theorem~\ref{Th.Na.uni-What} for single $t \in [0,\tau]$ holds in the Mallows metric $d_p$ for any even $0 < p \leq r$ as well.
  This of course includes white noise processes with $Poi(1)$ or standard normal margins, as applied later on.
}
\end{remark}

%
%

%
\section{Statistical Applications}
\label{sec:stat_app}
%

{\color{black}
Throughout this section denote by $\alpha \in (0,1)$ the nominal level of all inference procedures.
}

%
\subsection{Confidence Bands}
%

\label{sec:constrCB}

After having established all required weak convergence results,
we discuss different possibilities for realizing confidence bands for $A_j$ around the Nelson-Aalen estimator $\hat A_{jn}$, $j=1,\ldots,k,$
on an interval $[t_1, t_2] \subset [0,\tau]$ of interest. Later on, we propose a confidence band for differences of cumulative hazard functions.
As in Section~\ref{sec:brownian_bridge}, we first focus on $A_1$ and suppress the index 1 for notational convenience. Following \cite{abgk93}, Section~IV.1, we consider weight functions 
\begin{eqnarray*}
 g_1(s) = (s(1-s))^{-1/2} \qquad \text{or} \qquad g_2 \equiv 1
\end{eqnarray*}
as choices for $g$ in relation (\ref{eq:asymptdist}). 
The resulting confidence bands are commonly known as \emph{equal precision} and \emph{Hall-Wellner} bands, respectively.
We apply a log-transformation in order to improve small sample level $\alpha$ control.
Combining the previous sections' convergences with the functional delta-method and Slutsky's lemma yields
\begin{theorem} \label{thm:delta_meth_conv}
{\color{black}
Under condition \eqref{Pre.ass11},
}
for any $0 \leq t_1 \leq t_2 \leq \tau$ such that $A(t_1) > 0$, 
we have the following convergences in distribution on the c\`adl\`ag space $\mathfrak D[t_1,t_2]$:
 \begin{eqnarray}
\label{eq:weak_conv_logA.1}
 \Big( \sqrt n \hat A_n \frac{ \log \hat A_n - \log A}{1 + \hat\sigma^2} \Big) \cdot g \circ \frac{\hat\sigma^2}{1 + \hat\sigma^2}
 & \stackrel{d}{\longrightarrow} &(g B^0) \circ \phi \quad \text{and}\\
 \label{eq:weak_conv_logA.2}
  \Big( \frac{\hat W_n}{1 + \sigma^{*2}} \Big) \cdot g  \circ \frac{\sigma^{*2}}{1 + \sigma^{*2}} & \stackrel{d}{\longrightarrow}&
  (g B^0) \circ \phi
 \end{eqnarray}
conditionally given {\color{black}$\mathcal{C}_0$} in probability, 
with $\phi$ as in Section~\ref{sec:main} and the wild bootstrap variance estimator $\sigma^{*2}(s):=n\int_{(0,t]}J(s)Y^{-2}(s)G^2(s)$ $dN(s)$.
\end{theorem}
In particular, $\sigma^{*2}$ is a uniformly consistent estimate for $\sigma^2$ \citep{dobler14}
{\color{black}and, being the optional variation process of the wild bootstrap Nelson-Aalen process, it is a natural choice of a variance estimate}. 
For practical purposes, we adapt the approach of \citet{beyersmann12b} and estimate $\sigma^2$ based on the empirical variance of the wild bootstrap quantities $\hat{{W}}_{n}$. 
The continuity of the supremum functional translates \eqref{eq:weak_conv_logA.1} and~\eqref{eq:weak_conv_logA.2} 
into weak convergences for the corresponding suprema.
Hence, the consistency of the following critical values is ensured:
\begin{eqnarray*}
 c_{1-\alpha}^g & = & (1-\alpha) \text{ quantile of} \quad 
    \mathfrak L \Big( \sup_{s \in [t_1,t_2]} | g(\hat \phi(s)) B^0(\hat \phi(s)) | \Big), \\
 \tilde c_{1-\alpha}^g & = & (1-\alpha) \text{ quantile of} \quad 
    \mathfrak L \Big( \sup_{s \in [t_1,t_2]} \Big| \frac{\hat W_n(s)}{1 + \sigma^{*2}(s)} g \Big(\frac{\sigma^{*2}(s)}{1 + \sigma^{*2}(s)} \Big) \Big| \  \Big| \ {\color{black}\mathcal{C}_0}  \Big),
\end{eqnarray*}
where $\mathfrak L(\cdot)$ denotes the law of a random variable and $\alpha\in(0,1)$ the nominal level. 
Here, $g$ equals either $g_1$ or $g_2$ 
and $\hat\phi = \frac{\hat \sigma^2}{1+\hat \sigma^2}$. 
Note, that $\tilde c_{1-\alpha}^g$ is, in fact, a random variable. The results are back-transformed into four confidence bands for $A$ 
abbreviated with $HW$ and $EP$ for the Hall-Wellner and equal precision bands
and $a$ and $w$ for bands based on quantiles of the asymptotic distribution and the wild bootstrap, respectively.
In our simulation studies these bands are also compared with the linear confidence band $CB_{dir}^w$, which is based on the critical value
\begin{eqnarray*}
  \quad \tilde c_{1-\alpha} & = & (1-\alpha) \text{ quantile of} \quad 
    \mathfrak L \Big( \sup_{s \in [t_1,t_2]} \big| \hat W_n(s) \big| \Big| \ {\color{black}\mathcal{C}_0}  \Big).
\end{eqnarray*}

\begin{corollary}\label{cor:CBs}
Under the assumptions of Theorem \ref{thm:delta_meth_conv},
the following bands for the cumulative hazard function $(A(s))_{s \in [t_1,t_2]}$ provide an asymptotic coverage probability of $1-\alpha$:
 \begin{eqnarray}
 CB_{EP}^a & = & \Big[\hat A_n(s) \exp \Big( \mp \frac{c_{1-\alpha}^{g_1}}{\sqrt n \hat A_n(s)} \hat\sigma_n(s) \Big)\Big]_{s \in [t_1,t_2]} \nonumber \\
 CB_{HW}^a & = &\Big[\hat A_n(s) \exp \Big( \mp \frac{c_{1-\alpha}^{g_2}}{\sqrt n \hat A_n(s)} (1+\hat\sigma_n^2(s)) \Big)\Big]_{s \in [t_1,t_2]} \nonumber\\
 CB_{EP}^w & = &\Big[\hat A_n(s) \exp \Big( \mp \frac{\tilde c_{1-\alpha}^{g_1}}{\sqrt n \hat A_n(s)} \hat\sigma_n(s) \Big)\Big]_{s \in [t_1,t_2]} \label{eq:CBs} \\
 CB_{HW}^w & = & \Big[\hat A_n(s) \exp \Big( \mp \frac{\tilde c_{1-\alpha}^{g_2}}{\sqrt n\hat A_n(s)} (1+\hat\sigma_n^2(s)) \Big)\Big]_{s \in [t_1,t_2]} \nonumber \\
 CB_{dir}^w &= &\Big[\hat A_n(s) \mp \frac{\tilde c_{1-\alpha}}{\sqrt n}\Big]_{s \in [t_1,t_2]}. \nonumber
 \end{eqnarray}
\end{corollary}
\begin{remark}\label{rem:WB}

\begin{enumerate}
 \item Note that the wild bootstrap quantile $\tilde c_{1-\alpha}$ does not require an estimate of $\phi$, 
  thereby eliminating one possible cause of inaccuracy within the derivation of the other bands.
  However, the corresponding band $CB_{dir}^w$ has the disadvantage to possibly include negative values.
  \item The confidence bands are only well-defined 
  if the left endpoint $t_1$ of the bands' time interval is larger than the first observed event.
  In particular, these bands yield unstable results for small values of $\hat A_n(t_1)$
  due to the division in the exponential function;
  see \cite{lin1994} for a similar observation.
  \item The present approach directly allows the construction of confidence bands for within-sample comparisons of multiple $A_1, \dots, A_k$.
    For instance, a confidence band for the difference $A_1 - A_2$ may be obtained via quantiles based on the conditional convergence in distribution $\hat W_{1n} - \hat W_{2n} \stackrel{d}{\longrightarrow} U_1 - U_2 \sim Gauss( 0, \psi_1 + \psi_2)$ in probability by
    simply applying the continuous mapping theorem and taking advantage of the independence of $U_1$ and $U_2$;
    see \cite{whitt80} for the continuity of the difference functional.
    For that purpose, the distribution of
    \begin{eqnarray}
     D(t)=\sqrt{n} g(t)(\hat A_{1n}(t)-A_1(t)-(\hat A_{2n}(t)-A_2(t))),
    \end{eqnarray}
with positive weight function $g$  can be approximated by the conditional distribution of 
$\hat D(t)=  g(t) (\hat W_{1n}(t)-\hat W_{2n}(t))$.
With $g\equiv 1$, 
an approximate $(1-\alpha)\cdot 100\%$ confidence band for the difference $A_1 - A_2$ of two cumulative hazard functions on $[t_1,t_2]$ is 
\begin{eqnarray}
\left[\left(\hat A_1(s)-\hat A_2(s)\right)\pm \tilde q_{1-\alpha} / \sqrt{n}\right]_{s\in[t_1,t_2]}, \label{eq:diffCB}
\end{eqnarray}
where 
\begin{eqnarray*}
 \quad \tilde q_{1-\alpha} & = (1-\alpha) \text{ quantile of} \quad 
    \mathfrak L \Big( \sup_{s \in [t_1,t_2]} \big|\hat W_{1n}(s)-\hat W_{2n}(s) \big| \Big| \ {\color{black}\mathcal{C}_0}  \Big).
\end{eqnarray*}
Similar arguments additionally enable common two-sample comparisons. A practical data analysis using other weight functions $g$ in the context of cumulative incidence functions is given in \cite{hieke2013}.
\end{enumerate}
\end{remark}

\begin{remark}[Construction of Confidence Intervals]
\label{rem:cis}

\begin{enumerate}
 \item In particular,
Theorem~\ref{thm:delta_meth_conv} yields a convergence result on $\mathbb{R}^m$
for a finite set of time points $\{s_1, \dots, s_m \}\subset [0,\tau] , m \in \mathbb N$.
Hence, using critical values $\tilde c_{1-\alpha}$ and $\tilde c_{1-\alpha}^g$
obtained from the law of the maximum $\max_{s_1, \dots, s_m}$ instead of the supremum,
 a variant of Corollary~\ref{cor:CBs} specifies simultaneous confidence intervals $I_1 \times \dots \times I_m$ 
 for $(A(s_1), \dots, A(s_m))$ with asymptotic coverage probability $1-\alpha$.
 Since the error multiplicity is taken into account,
 the asymptotic coverage probability of a single such interval $I_j$ for $A(s_j)$ is greater than $1-\alpha$.
 \item Due to the asymptotic independence of the entries of the multivariate Nelson-Aalen estimator,
a confidence region for the value of a multivariate cumulative hazard function $(A_1(t), \dots, A_k(t))$ at time $t \in [0,\tau]$
may be found using \v{S}id\'ak's correction:
Letting $J_1, \dots, J_k$ be pointwise confidence intervals for $A_1(t),$ $\dots,$ $A_k(t)$ with asymptotic coverage probability $(1-\alpha)^{1/k}$,
each found using the wild bootstrap principle,
the coverage probability of $J_1 \times \dots \times J_k$ for $A_1(t) \times \dots \times A_k(t)$ clearly goes to $1 -\alpha$ as $n \rightarrow \infty$.
\end{enumerate}
\end{remark}

%
\subsection{Hypothesis Tests for Equivalence{\color{black}, Inferiority, Superiority,} and Equality}
%
Adapting the principle of confidence interval inclusion as discussed in \cite{wellek10}, Section~3.1, to time-simultaneous confidence bands,
hypothesis tests for equivalence of cumulative hazard functions become readily available.
To this end, let $\ell, u: [t_1,t_2] \rightarrow (0,\infty)$ be positive, continuous functions 
and denote by $(a_n(s),\infty)_{s \in [t_1,t_2]}$ and $[0,b_n(s))_{s \in [t_1,t_2]}$ the one-sided (half-open) analogues of any confidence band of the previous subsection with asymptotic coverage probability $1-\alpha$.
Furthermore, let $A_0: [t_1,t_2] \rightarrow [0,\infty)$ be a pre-specified non-decreasing, continuous function for which equivalence to $A$ shall be tested. More precisely:
\begin{eqnarray*}
 H : \{ A(s) \leq A_0(s) - \ell(s) \text{ or } A(s) \geq A_0(s) + u(s) \text{ for some } s \in [t_1,t_2] \} \\
 \text{vs.} \quad
 K : \{ A_0(s) - \ell(s) < A(s) < A_0(s) + u(s) \text{ for all } s \in [t_1,t_2] \}.
\end{eqnarray*}
\begin{corollary}
\label{cor:equivalence}
 Under the assumptions of Theorem \ref{thm:delta_meth_conv}, a hypothesis test $\psi_n$ of asymptotic level $\alpha$ for $H$ vs $K$ is given by the following decision rule: 
Reject $H$ if and only if the combined two-sided confidence band $(a_n(s),b_n(s))_{s \in [t_1,t_2]}$ 
is fully contained in the region spanned by $(A_0(s) - \ell(s), A_0(s) + u(s))_{s \in [t_1,t_2]} $.
Further, it holds under $K$ that $\E(\psi_n) \rightarrow 1$ as $n\rightarrow \infty$, i.e., $\psi_n$ is consistent.
\end{corollary}
{\color{black}Similar arguments lead to analogue one-sided tests for the inferiority or superiority of the true cumulative hazard function to a prespecified function $A_0$.}
Moreover, statistical tests for equality of two cumulative hazard functions can be constructed
using the weak convergence results of Remark~\ref{rem:WB}(c):
$$H_{=} : \{ A_1 \equiv A_2 \text{ on } [t_1,t_2] \} \quad \text{vs} \quad K_{\neq} : \{ A_1(s) \neq A_2(s) \text{ for some } s \in [t_1,t_2] \}.$$
Corollary~\ref{cor:ks_test} below yields an asymptotic level $\alpha$ test for $H_{=}$.
\cite{bajorunaite07} and \cite{dobler14} used similar two-sided tests for comparing cumulative incidence functions in a two-sample problem.
\begin{corollary}[A Kolmogorov-Smirnov-type test]
\label{cor:ks_test}
  Under the assumptions of Theorem~\ref{thm:delta_meth_conv} and letting $g$ again be a positive weight function,
  $$\varphi^{KS}_n = \mathbf 1\{ \sup_{s \in [t_1,t_2]} \sqrt{n} g(s) | \hat A_{1n}(s) - \hat A_{2n}(s) | > \tilde q_{1-\alpha} \}$$
 defines a consistent, asymptotic level $\alpha$ resampling test for $H_=$ vs. $K_{\neq}$.
 Here $\tilde q_{1-\alpha}$ is the $(1-\alpha)$-quantile of 
  $\mathfrak L \big( \sup_{s \in [t_1,t_2]} \big|\hat D(s) \big| \big| \ {\color{black}\mathcal{C}_0}  \big)$.
\end{corollary}
Similarly, Theorem~\ref{thm:delta_meth_conv} enables the construction of other tests, e.g., such of Cram\'er-von Mises-type.
Furthermore, by taking the suprema over a discrete set $\{s_1,\dots, s_m\} \subset [0,\tau]$,
the Kolmogorov-Smirnov test of Corollary~\ref{cor:ks_test} can also be used to test
\begin{align*}
 & \tilde H_{=} : \{ A_1(s_j) = A_2(s_j) \text{ for all } 1\le j\leq m \} \\
 & \quad \text{vs.} \quad \tilde K_{\neq} : \{ A_1(s_j) \neq A_2(s_j) \text{ for some }1\le j\leq m \} .
\end{align*}
Note that in a similar way, two-sample extensions of Corollaries \ref{cor:equivalence} and \ref{cor:ks_test} can be established following \cite{dobler14}.

{\color{black}
\subsection{Tests for Proportionality}\label{sec:proptest}

A major assumption of the widely used \cite{cox72} regression model is the assumption of proportional hazards over time. 
Several authors have developed procedures for testing the null hypothesis of proportionality, see e.g. \cite{gill1987simple}, \cite{lin1991goodness}, \cite{grambsch1994proportional}, \cite{hess1995graphical}, \cite{scheike2004estimation} or \cite{kraus2007data} and the references cited therein.
We apply our theory to derive a non-parametric test for proportional hazards assumption of two samples in our very general framework, 
covering two-sample right-censored and left-truncated multi-state models.
The framework is an unpaired two-sample model given by independent counting processes $N^{(1)}, N^{(2)}$ and predictable processes $Y^{(1)}, Y^{(2)}$,
assuming the conditions of Section~\ref{sec:model} for each group,
and with sample sizes $n_1$ and $n_2$, respectively.
Let again $J^{(j)}(t) = \mathbf 1\{ Y^{(j)}(t) > 0 \}$, $j=1,2$.
Denote by $\hat A^{(j)}_{n_j} = \int_{(0,t]} \frac{J^{(j)}(s)}{Y^{(j)}(s)} d N^{(j)}$ the Nelson-Aalen estimator 
of the cumulative hazard functions $A^{(j)}$ and by $\alpha^{(j)}$ the corresponding rates, $j=1,2$.
To motivate a suitable test statistic we make use of the following equivalence between hazards proportionality and equality of both cumulative hazards: 
$$
\alpha^{(1)}(t) = c \ \alpha^{(2)}(t) \ \text{in} \ t \in [0,\tau] \  \text{for} \ c > 0 
\ \ \Longleftrightarrow \ \
A^{(1)}(t) = c \ A^{(2)}(t) \ \text{in} \ t \in [0,\tau] \  \text{for} \ c > 0 ,
$$
which, as the null hypothesis of interest, is denoted by $H_{0,\text{prop}}$.
In a natural way similar to \cite{gill1987simple} in the simple survival setup this leads to statistics of the form
\begin{align*}
T_{n_1,n_2} = \rho \Big( \ \sqrt{\frac{n_1 n_2}{n}} \frac{\hat A^{(2)}_{n_2}}{\hat A^{(1)}_{n_1}} \ , \ \sqrt{\frac{n_1 n_2}{n}} \frac{\hat A^{(2)}_{n_2}(\tau)}{\hat A^{(1)}_{n_1}(\tau)} \ \Big), 
\end{align*}
$n=n_1+n_2$, where $\rho$ is an adequate distance on $\mathfrak D[0,\tau]$, e.g. $\rho(f,g) = \sup w |f - g|$ (leading to Kolmogorov-Smirnov-type tests),
\ $\rho(f,g) = \int (f-g)^2 w^2 d \lambda \!\! \lambda$
(leading to Cram\'er-von-Mises-type tests), 
where $w: [0,\tau] \rightarrow [0,\infty)$ is a suitable weight function.
Later on, we choose $w = \hat A^{(1)}_{n_1}$ which ensures the evaluation of $\rho$ on $\{ \hat A^{(1)}_{n_1} > 0 \}$.
Let $\hat W_{n_1}^{(1)}$ and $\hat W_{n_2}^{(2)}$ be the obvious wild bootstrap versions of the sample-specific centered Nelson-Aalen estimators;
cf. \eqref{NA.uni.Wnh}.

\begin{theorem}
\label{thm:prop}
 Let $\rho$ be either the above Kolmogorov-Smirnov- or the Cram\'er-von Mises-type statistic with $w = \hat A^{(1)}_{n_1}$.
 If $n_1 / n \rightarrow p \in (0,1)$ as $\min(n_1,n_2) \rightarrow \infty$,
 then the test for $H_{0,\textnormal{prop}}$
 $$ \varphi_{n_1,n_2}^{\textnormal{prop}} = \mathbf 1 \{ T_{n_1,n_2} > \tilde q_{1 - \alpha} \} $$
 has asymptotic level $\alpha$ under $H_{0,prop}$ and asymptotic power 1 on the whole complement of $H_{0,prop}$.
 Here $\tilde q_{1 - \alpha}$ is the $(1 - \alpha)$-quantile of 
 $$\mathfrak L \Big(
 \rho\Big( \sqrt{\frac{n_1}{n}}  \frac{\hat W_{n_2}^{(2)}}{\hat A_{n_1}^{(1)}} - \sqrt{\frac{n_2}{n}}  \hat W_{n_1}^{(1)} \frac{\hat A_{n_2}^{(2)}}{[\hat A_{n_1}^{(1)}]^2} \ , \
 \sqrt{\frac{n_1}{n}} \frac{ \hat W_{n_2}^{(2)}(\tau)}{\hat A_{n_1}^{(1)}(\tau)} - \sqrt{\frac{n_2}{n}}  \hat W_{n_1}^{(1)}(\tau) \frac{\hat A_{n_2}^{(2)}(\tau)}{[\hat A_{n_1}^{(1)}(\tau)]^2} \Big) 
 \Big| \ \mathcal{C}_0 \Big).$$ 
\end{theorem}
%
%
%

}

%
\section{Simulation Study}
%
\label{sec:simus}
The motivating example behind the present simulation study is the SIR-3 data of Section \ref{sec:dataex}. The setting is a specification of Example \ref{ex:cov_models}(a) called {\it illness-death model  with recovery}. As illustrated in the multistate pattern of Figure \ref{fig:illnessdeath}, the model has state space  $\mathcal S=\lbrace 0,1,2\rbrace$ and includes the transition hazards $\alpha_{01},\ \alpha_{10},\ \alpha_{02},$ and $\alpha_{12}$.
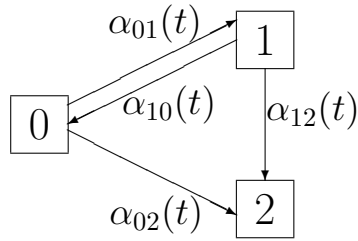
\begin{figure}[ht]
  \begin{center}
  \setlength{\unitlength}{0.75cm}
 \makebox{ \begin{picture}(5,6)(-1,-3)
    
    \put(-2,0){\framebox(1,1){\Large 0}}
    
    \put(-1,0.85){\vector(2,1){3}}\put(2,1.5){\framebox(1,1){\Large 1}}
    \put(-1,0.4){\vector(2,-1){3}}\put(2,-1.5){\framebox(1,1){\Large 2}}
    \put(2.5,1.5){\vector(0,-11){2}}

    \put(2,2){\vector(-2,-1){3}}

    \put(0.05,1.8){\makebox(1,1){\large$\alpha_{01}(t)$}}
    \put(0.075,-1.6){\makebox(1,1){\large$\alpha_{02}(t)$}}
    \put(2.85,0.25){\makebox(1,1){\large$\alpha_{12}(t)$}}
    \put(0.3,0.35){\makebox(1,1){\large$\alpha_{10}(t)$}}

  \end{picture}}  
  \caption{Illness-death model with recovery and transition hazards $\alpha_{01},\ \alpha_{10},\ \alpha_{02},$ and $\alpha_{12}$ at time $t$.}
    \label{fig:illnessdeath}
  \end{center}
\end{figure}
The simulation of the underlying quantities is based on the methodology suggested by \cite{allignol11} generalized to the time-inhomogeneous Markovian multistate framework, which can be seen as a nested series of competing risks experiments. More precisely, the individual initial states are derived from the proportions of individuals at $t=0$ and the censoring times are obtained from a multinomial experiment using probability masses equal to the increments of the censoring Kaplan-Meier estimate originated from the SIR-3 data. Similarly, event times are generated according to a multinomial distribution with probabilities given by the increments of the original Nelson-Aalen estimators. These times are subsequently included into the multistate simulation algorithm described in \cite{beyersmann12a}, Section 8.2. Since censoring times are sampled independently and each simulation step is only based on the current time and the current state, the resulting data follows a Markovian structure. A more formal justification of the multistate simulation algorithm can be found in \cite{gill1990survey} and Theorem II.6.7 in \cite{abgk93}.
\begin{table}
\caption{Mean number of events per transition on [5,30] provided by the simulation study of Section \ref{sec:simus}}
\label{tab:nevents}
\centering{\fbox{
\begin{tabular}{lllll}
\multirow{2}{*}{Sample size} & \multicolumn{4}{c}{Transition}\\ \cline{2-5} 
& $1\rightarrow 0$   & $0\rightarrow 1$   & $0\rightarrow 2$   & \multicolumn{1}{c}{$1\rightarrow 2$} \\ \hline
93  & 20.1 &4.3&   43.9    & 10.1   \\
186 & 42.7 & 8.3  & 96.5  & 21.7\\
373 & 85.6 & 17.0 & 193.4 & 43.7\\
747 & 170.9 & 33.9  & 387.4 & 87.4\\
747*& 171   & 34    &  387     & 87\\ \hline *original data	
\end{tabular}}}

\end{table}
We consider three different sample sizes: The original number of 747 patients is stepwisely reduced to {\color{black}373, 186, and 93 patients}. For each scenario we simulate 1000 studies. As an overview, the mean number of events for each possible transition and scenario is illustrated in Table 1. 

The mean number of events regarding 747 patients reflects the original number of events.
All numbers are restricted to the time interval [5,30], which is chosen due to a small amount of events before $t=5$ (left panel of Figure \ref{fig:CBdataex}). Further, less than 10\% of all individuals are still under observation after day 30. In particular, asymptotic approximations tend to be poor at the left- and right-hand tails; cf. Remark \ref{rem:WB}(b) and \cite{lin97}. 

Utilizing the \texttt{R}-package \texttt{sde} \citep{sde2014}, the quantiles $c_{1-\alpha}^g$ in (\ref{eq:CBs}) of each single study are empirically estimated by simulating 1000 sample paths of a standard Brownian bridge. These quantiles are separately derived for both the Aalen-  and Greenwood-type variance estimates (\ref{eq:varaalen}) and (\ref{eq:vargreenwood}). The bootstrap critical values are based on 1000 bootstrap realizations of $\hat{\textbf{\textit{W}}}_n$ for each simulation step including both standard normal and centered Poisson variates with variance one. The latter is motivated by a slightly better performance compared to standard normal multipliers (\citealp{beyersmann12b}, and \citealp{dobler15a}).
{\color{black}
 Furthermore, \cite{liu88} argued in a classical (linear regression) problem that wild bootstrap weights with skewness equal to one
 satisfy the second order correctness of the resampling approach.
 According to the cited simulation results, a similar result might hold true in our context,
 as the Poisson variates have skewness equal to one and standard normal variates are symmetric.
 A careful analysis of the convergence rates, however, is certainly beyond the scope of this article. In order to guarantee statistical reliability, we do not derive confidence bands for sample sizes and transitions with a mean number of observed transitions distinctly smaller than 20. The nominal level is set to $\alpha=0.05$. All simulations are performed with the \texttt{R}-computing environment version 3.3.2 \citep{Rcite}.
}

Following Table \ref{tab:CovProb}, {\color{black} almost all} bands constructed via Brownian bridges consistently tend to be rather conservative in our setting, i.e., result in too broad bands. Here, the usage of the Greenwood-type variance estimate yields more accurate coverage probabilities compared to the Aalen-type estimate. 
In contrast, the wild bootstrap approach {\color{black} mostly outperforms the Brownian bridge procedures: The log-transformed wild bootstrap bands approximately keep the nominal level even in the smaller sample sizes, except for the $0 \rightarrow 1$ transition with smallest sample size (corresponding to only 17 events in the mean; cf. Table~\ref{tab:nevents}).
 We also observe that the log-transformation in general improves coverage for the wild bootstrap procedure.
 The current simulation study showed no clear preference for the choice of weight.
 Note that all wild bootstrap bands for transition $0\rightarrow 2$ show a similar, but mostly reduced conservativeness compared to the bands provided by Brownian bridges.} We have to emphasize that coverage probabilities for the cumulative hazard functions are drastically decreased to approximately 75\% in all sample sizes if log-transformed pointwise confidence intervals would wrongly be interpreted time-simultaneously (results not shown).

\begin{sidewaystable}
\caption{Empirical coverage probabilities (\%) from the simulation study of Section \ref{sec:simus} separately for each transition and different simulated number of individuals}
\label{tab:CovProb}
\centering{\begin{tabular}{|lc|cccccccccc|}
\hline
\multicolumn{2}{|l|}{\multirow{4}{*}{}} & \multicolumn{10}{c|}{Type of Confidence Band}                                                                                                                                   \\ \cline{3-12} 
\multicolumn{2}{|l|}{}                  & \multicolumn{4}{c}{Brownian Bridge}                               & \multicolumn{6}{c|}{Wild Bootstrap}                                                                         \\ \cline{3-12} 
\multicolumn{2}{|l|}{}                  & \multicolumn{2}{c}{95\% log EP} & \multicolumn{2}{c}{95\% log HW} & \multicolumn{2}{c}{95\% log EP} & \multicolumn{2}{c}{95\% log HW}     & \multicolumn{2}{c|}{95\% direct}    \\ \cline{3-12} 
\multicolumn{2}{|l|}{}                  & Aalen          & Green-         & Aalen          & Green-         & Poisson        & Standard       & \multirow{2}{*}{Poisson} & Standard & \multirow{2}{*}{Poisson} & Standard \\
Transition                        & $N$ &                & wood           &                & wood           &                & normal         &                          & normal   &                          & normal   \\ \hline
\multirow{2}{*}{$0\rightarrow 1$} & 373 & 96.4           & 95.9           & 95.5           & 95.3           & 92.5           & 92.5           & 92.5                     & 92.5     & 91.4                     & 91.9     \\
                                  & 747 & 97.7           & 97.3           & 97.2           & 97.0           & 95.0           & 94.9           & 95.2                     & 95.0     & 92.9                     & 93.2     \\ \hline
\multirow{4}{*}{$0\rightarrow 2$} & 93  & 98.0           & 97.1           & 98.4           & 97.3           & 97.6           & 97.8           & 97.3                     & 96.0     & 96.6                     & 96.6     \\
                                  & 186 & 98.3           & 95.5           & 98.9           & 97.4           & 97.2           & 98.2           & 98.0                     & 96.1     & 96.2                     & 96.2     \\
                                  & 373 & 98.1           & 95.0           & 98.2           & 96.9           & 97.2           & 97.3           & 97.1                     & 97.1     & 96.0                     & 96.3     \\
                                  & 747 & 98.6           & 96.2           & 98.8           & 97.7           & 97.7          & 97.8           & 97.4                     & 97.8     & 96.0                     & 96.2     \\ \hline
\multirow{4}{*}{$1\rightarrow 0$} & 93  & 97.0           & 94.9           & 97.0           & 95.2           & 95.1           & 95.1           & 94.8                     & 94.8     & 93.7                     & 93.7     \\
                                  & 186 & 97.3           & 95.8           & 97.7           & 96.1           & 95.6           & 95.7           & 95.7                     & 95.4     & 94.5                     & 94.3     \\
                                  & 373 & 97.2           & 96.3           & 97.9           & 97.0           & 95.2           & 95.3           & 95.9                     & 96.3     & 95.2                     & 95.3     \\
                                  & 747 & 97.8           & 96.8           & 97.5           & 96.9           & 96.6           & 96.6           & 95.9                     & 96.0     & 96.1                     & 96.3     \\ \hline
\multirow{3}{*}{$1\rightarrow 2$} & 186 & 97.5           & 96.7           & 97.2           & 96.2           & 94.7           & 94.5           & 94.3                     & 94.7     & 93.2                     & 93.3     \\
                                  & 373 & 98.2           & 97.7           & 98.2           & 97.8           & 95.8           & 95.8           & 95.1                     & 95.2     & 94.6                     & 94.3     \\
                                  & 747 & 97.2           & 96.6           & 96.6           & 96.0           & 94.4           & 95.5           & 94.3                     & 94.7     & 94.9                     & 95.1     \\ \hline
\end{tabular}}
\flushleft{EP: equal-precision band; HW: Hall-Wellner band.}
\end{sidewaystable}

{\color{black}
The second set of simulations follows the test for proportional hazards derived in Theorem \ref{thm:prop} with regard to keeping the preassigned error level under the null hypothesis. For that purpose, we assume a competing risks model with two competing events separately for two unpaired patient groups. For an illustration, see for instance, Figure 3.1 in \cite{beyersmann12a}. 

We consider four different constant hazard scenarios: (I) the hazards for the type-1 event are set to  $\alpha^{(1)}_{01}(t)=\alpha^{(2)}_{01}(t)=2$ (no effect on the type-1 hazard, in particular, a hazard ratio of $c=1$); (II) $\alpha^{(1)}_{01}(t)=1$ and $\alpha^{(2)}_{01}(t)=2$ (large effect); (III) $\alpha^{(1)}_{01}(t)=\alpha^{(2)}_{01}(t)=1$; (IV) $\alpha^{(1)}_{01}(t)=1$ and $\alpha^{(2)}_{01}(t)=1.5$ (moderate effect). In each scenario, we set $\alpha^{(1)}_{02}=\alpha^{(2)}_{02}(t)=2$, in particular, we consistently assume no group effect on the competing hazard. Further, scenario-specific administrative censoring times are chosen such that approximately 25\% of the individuals are censored. The simulations designs are selected such that we include different effect sizes as well as different type-1 hazard ratio configurations with respect to the competing hazards. We consider a balanced design with $n_1=n_2=n\in\lbrace 125,250,500,1000\rbrace$. The right-hand tail of the domain of interest is set to $\tau=0.3$. Simulation of the event times and types follows the procedure explained in Chapter 3.2 of \cite{beyersmann12a}. As before, we simulate 1000 studies for each scenario and sample size configuration, whereas the critical values of the Kolmogorov-Smirnov-type and Cram\'er-von-Mises-type statistics from Section~\ref{sec:proptest} are derived from 1000 bootstrap samples including both standard normal and centered Poisson variates with variance one.

The results for the type I error rates (for $\alpha=0.05$) are displayed in Table \ref{tab:proptest}. As expected from consistency, the higher the number of patients the better is the type I error approached for both test statistics in each scenario. 
Except for Scenario (II), all procedures keep the type I error rate quite accurately for $n\geq 500$.
For smaller sample sizes, all tests tend to be conservative with a particular advantage for the Kolmogorov-Smirnov statistic.


 \begin{sidewaystable}[]
\centering
\caption{Simulated size of $\phi_{n_1,n_2}^{\text{prop}}$ for nominal size $\alpha=5\%$ under different sample sizes and constant hazard configurations. In each scenario $\tau=0.3$ and $25\%$ of individuals are censored.}

\begin{tabular}{|l|cccc|cccc|cccc|cccc|}
\hline
      & \multicolumn{4}{c|}{Scenario I}                                      & \multicolumn{4}{c|}{Scenario II}                                     & \multicolumn{4}{c|}{Scenario III}                                    & \multicolumn{4}{c|}{Scenario IV}                                     \\ \cline{2-17} 
      & \multicolumn{2}{c|}{KMS}         & \multicolumn{2}{c|}{CvM} & \multicolumn{2}{c|}{KMS}         & \multicolumn{2}{c|}{CvM} & \multicolumn{2}{c|}{KMS}         & \multicolumn{2}{c|}{CvM} & \multicolumn{2}{c|}{KMS}         & \multicolumn{2}{c|}{CvM} \\ \hline
$n_i$ & SN    & \multicolumn{1}{l|}{Poi} & SN          & Poi        & SN    & \multicolumn{1}{l|}{Poi} & SN          & CvM        & SN    & \multicolumn{1}{l|}{Poi} & SN          & Poi        & SN    & \multicolumn{1}{l|}{Poi} & SN         & Poi         \\ \hline
125   &   0.029    &     0.024                     & 0.033            &    0.030        & 0.029      &           0.026               & 0.027            &  0.023          & 0.045 & 0.041                    & 0.028       & 0.030      & 0.046 & 0.042                    & 0.030      & 0.025       \\
250   & 0.035 & 0.039                    & 0.039       & 0.040      & 0.040 & 0.038                    & 0.037       & 0.034      & 0.039 & 0.040                    & 0.037       & 0.034      & 0.034 & 0.034                    & 0.033      & 0.030       \\
500   & 0.057 & 0.054                    & 0.059       & 0.060      & 0.034 & 0.038                    & 0.040       & 0.041      & 0.056 & 0.053                    & 0.047       & 0.045      & 0.044 & 0.044                    & 0.043      & 0.044       \\
1000  &  0.050     &     0.050                     &    0.047         &    0.047        &  0.048     &     0.049                     &    0.043         &        0.046    &  0.047     &        0.049                  &      0.045       &   0.048         &  0.058    &        0.059                  &  0.053          &      0.056       \\ \hline
\end{tabular}
\flushleft{KMS: Kolmogorov-Smirnov-type statistic, CvM: Cram\'er-von-Mises-type statistic; SN: standard normal multiplier; Poi: centered poisson multiplier.}
\label{tab:proptest}
\end{sidewaystable}
  
}
\section{Data Example}\label{sec:dataex}
%
The SIR-3 (\emph{S}pread of Nosocomial \emph{I}nfections
and \emph{R}esistant Pathogens) cohort study at the Charit{\'e} University
Hospital in Berlin, Germany, prospectively collected data on the occurrence
and consequences of hopital-aquired infections in intensive care
\citep{beyersmann2006,wolkewitz2008}. A device of particular interest in
critically ill patients is mechanical ventilation. The present data analysis
investigates the impact of ventilation on the length of intensive care unit
stay which is, e.g., of interest in cost containment analyses in hospital
epidemiology \citep{beyersmann2011}. The analysis considers a random subset of 747
patients of the SIR-3 data which one of us has made publicly available
\citep{beyersmann12a}. Patients may either be ventilated
(state~1 as in Figure~\ref{fig:illnessdeath}) or not ventilated (state~0)
upon admission. Switches in device usage are modelled as transitions between
the intermediate states~0 and~1. Patients move into state~2 upon discharge
from the unit. The numbers of observed transitions are reported in the last
row of Table 1. We start by separately considering the two
cumulative end-of-stay hazards~$A_{12}$ and $A_{02}$, followed by a more
formal group comparison as in Remark \ref{rem:WB}(c). Based on the approach suggested by \citet{beyersmann12a}, Section 11.3, we find it reasonable to assume the Markov property.
\begin{figure}

 	 \makebox{ \includegraphics[width=0.9\textwidth]{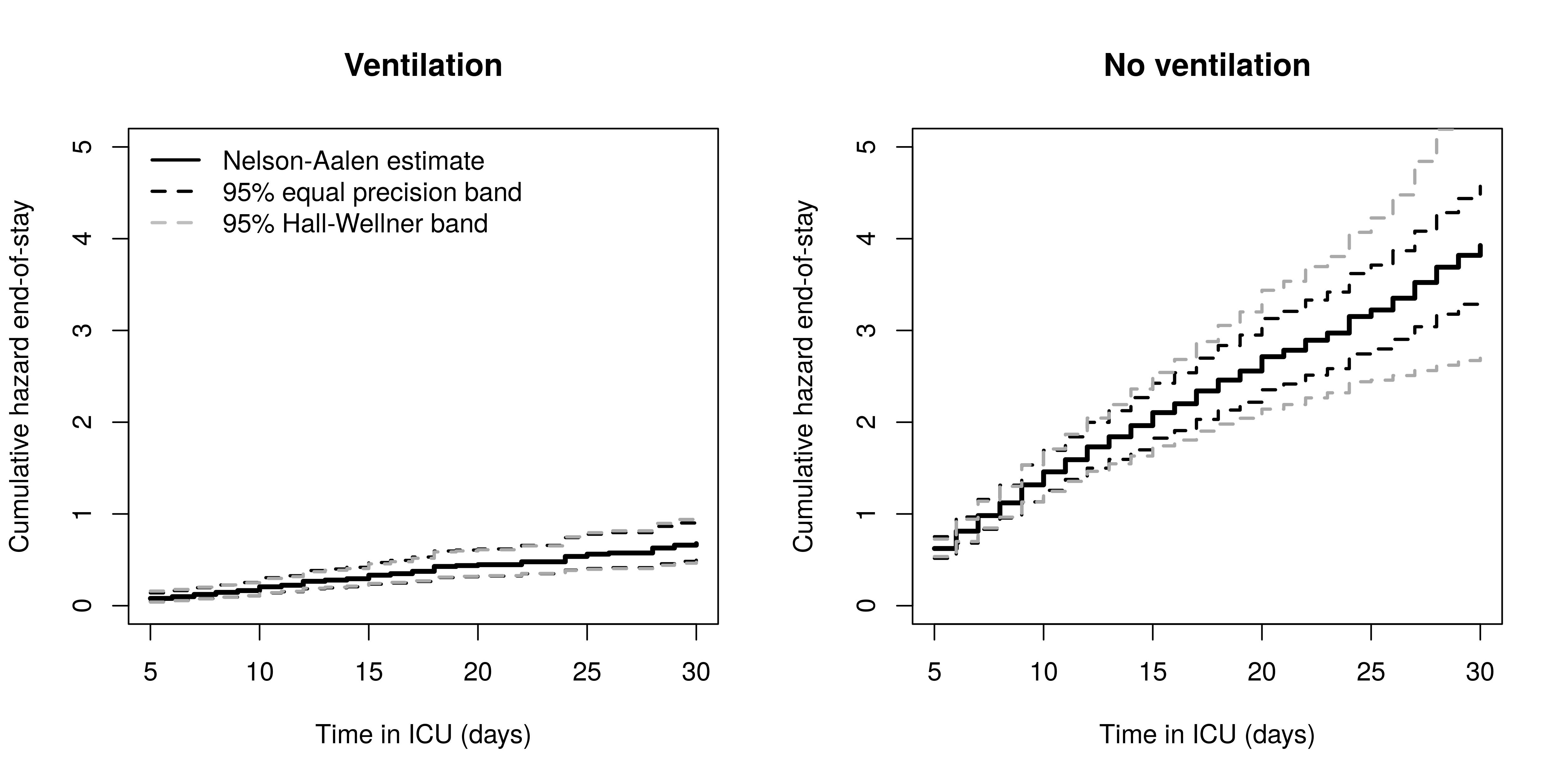}}
	\caption{95\% confidence bands based on standard normal multipliers for the cumulative hazard of end-of-stay from the data example in Section \ref{sec:dataex}. The solid black lines are the Nelson-Aalen estimators separately for `no ventilation' (state 0, right plot) and `ventilation' (state 1, left plot).}
	\label{fig:CBdataex}
\end{figure}
Figure \ref{fig:CBdataex} displays the Nelson-Aalen estimates
of~$A_{12}$ and $A_{02}$ accompanied by simultaneous 95\% confidence bands
utilizing the 1000 wild bootstrap versions with standard normal variates and
restricted to the time interval [5,30] of intensive care unit days. As before, the left-hand tail of the interval is chosen, because Nelson-Aalen estimation regarding $A_{12}$ picks up at $t=5$, cf. the left panel of Figure \ref{fig:CBdataex}. \textcolor{black}{Graphical validation of empirical means and variances of $\hat{\boldsymbol W}_{n}$ showed good compliance compared to the theoretical limit quantities stated in Remark \ref{rm:wc}.} Bands using Poisson variates are similar (both results not shown). 
\begin{figure}
\centering
 	 \makebox{\includegraphics[width=0.9\textwidth]{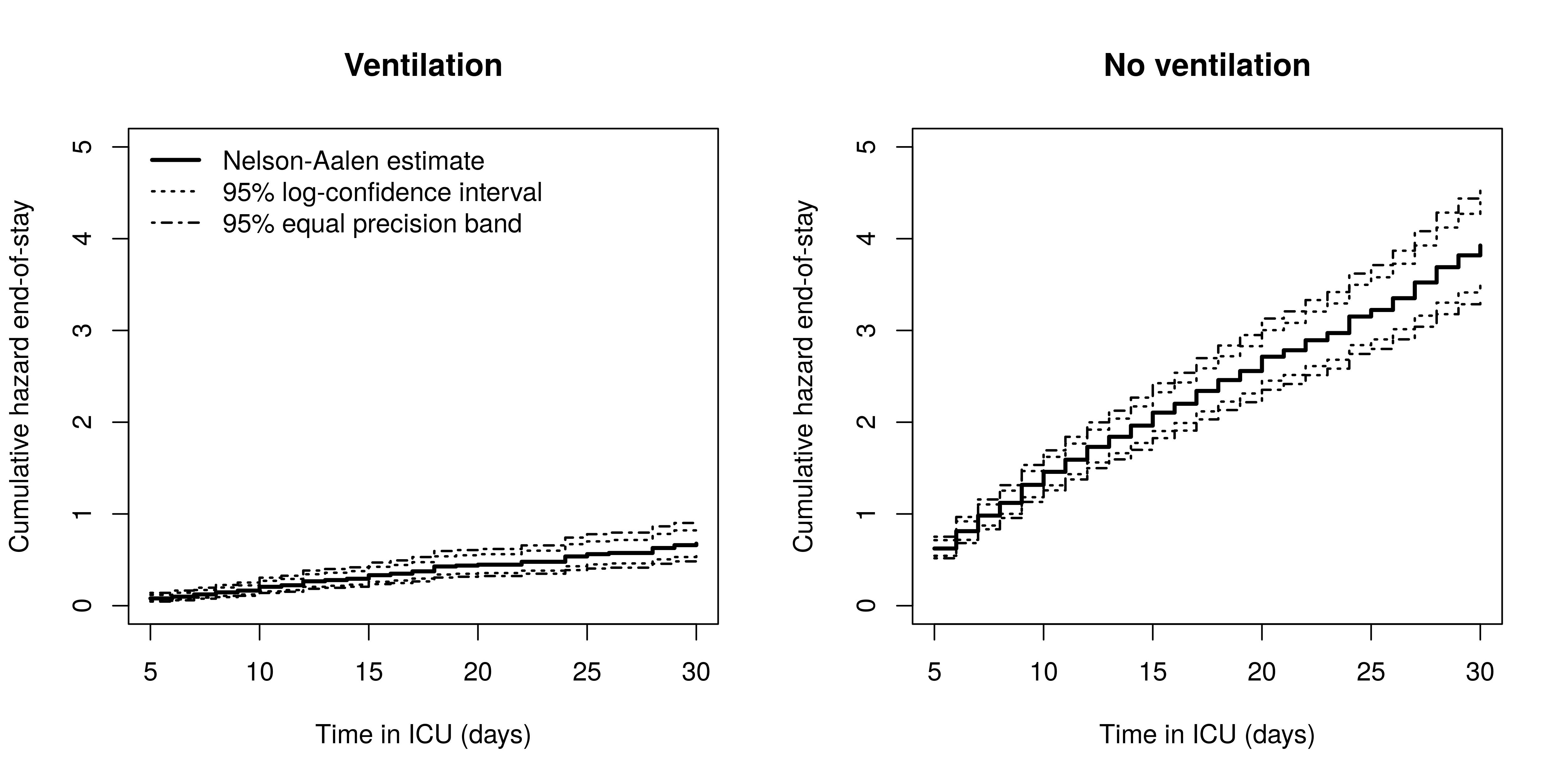}}
	\caption{95\% equal precision confidence bands based on standard normal multipliers and 95\% log-transformed pointwise confidence intervals for the cumulative hazard of end-of-stay from the data example in Section \ref{sec:dataex}. The solid black lines are the Nelson-Aalen estimators separately for `no ventilation' (state 0, right plot) and `ventilation' (state 1, left plot).}
	\label{fig:CBCI}
\end{figure}
Figure \ref{fig:CBCI} also displays the 95\% pointwise confidence intervals based
on a log-transformation. The performance of both equal precision and
Hall-Wellner bands is comparable for transitions out of the ventilation
state. However, the latter tend to be larger for the $0\rightarrow 2$
transitions for later days due to more unstable weights at the right-hand
tail. Equal precision bands are graphically competitive when compared to the
pointwise confidence intervals. Ventilation significantly reduces the hazard
of end-of-stay, since the upper half-space is not contained in the 95\% confidence band of the cumulative hazard difference, see Figure \ref{fig:DiffCB}.


\begin{figure}
\centering
 	 \makebox{\includegraphics[width=0.5\textwidth]{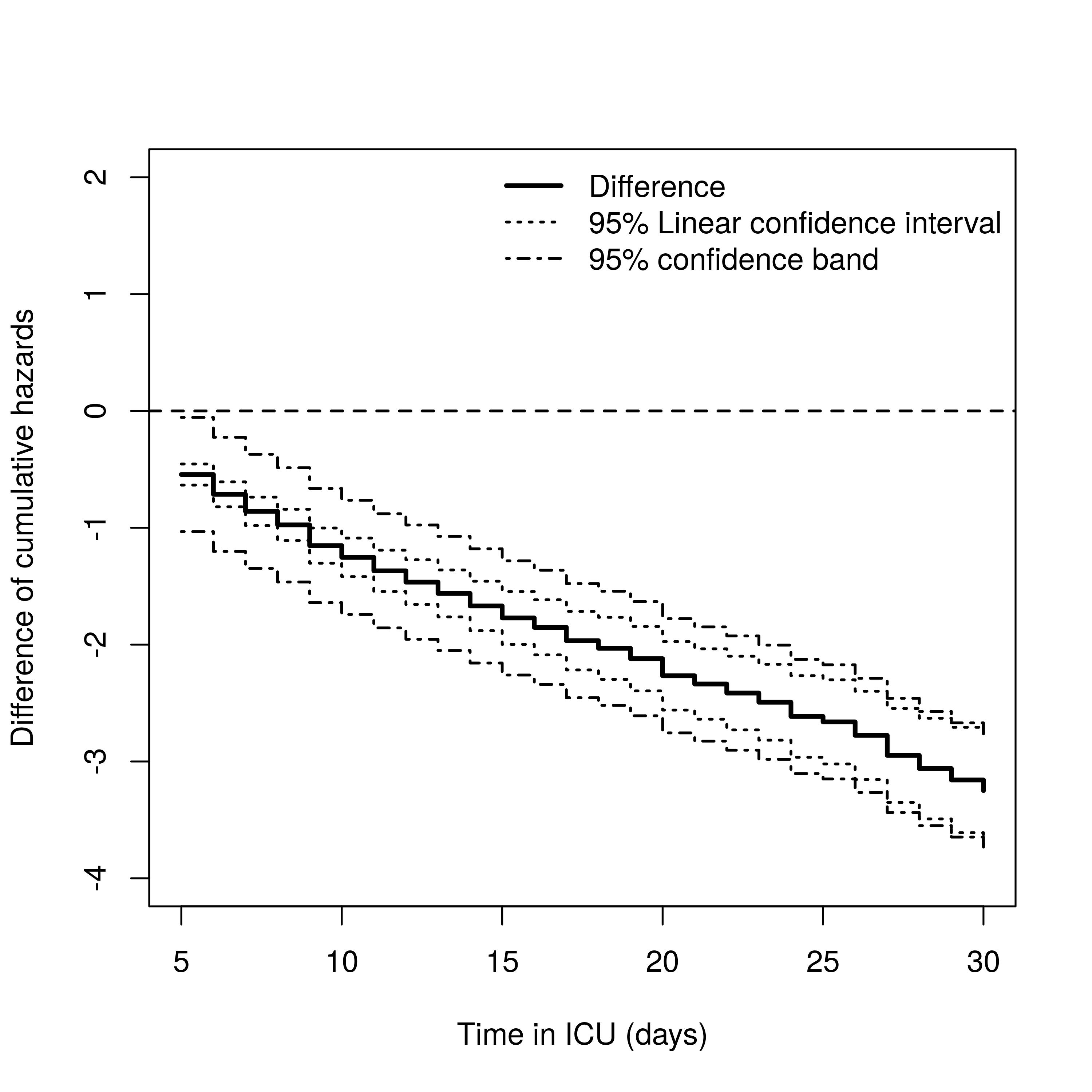}}
	\caption{95\% confidence bands from relation \eqref{eq:diffCB} based on standard normal multipliers and 95\% linear pointwise confidence intervals for difference of the two cumulative hazards of end-of-stay from the data example in Section \ref{sec:dataex}. The solid black lines is the difference `ventilation vs. no ventilation' of the Nelson-Aalen estimators within the two ventilation groups.}
	\label{fig:DiffCB}
\end{figure}
%
\section{Discussion and further Research}
\label{sec:discussion}
%
We have given a rigorous presentation of a weak convergence result for the wild bootstrap methodology for the multivariate Nelson-Aalen estimator in a general setting only assuming Aalen's multiplicative intensity structure of the underlying counting processes. This allowed the construction of time-simultaneous confidence bands and intervals as well as asymptotically valid equivalence and equality tests for cumulative hazard functions. In the context of time-to-event analysis, our general framework is not restricted to the standard survival or competing risks setting, but also covers arbitrary Markovian multistate models with finite state space, other classes of intensity models like relative survival or excess mortality models, and even specific semi-Markov situations. Additionally, independent left-truncation and right-censoring can be incorporated. {\color{black} The procedure has also been used to construct a test for proportional hazards.}
Easy and computationally convenient implementation and within- or two-sample comparisons demonstrate its attractiveness in various practical applications. 

Future work will be on the approximation of the asymptotic distribution corresponding to the matrix of transition probabilities (see \citealp{AalenJoh78}) \textcolor{black}{and functionals thereof in general Markovian multistate models. This is of great practical interest, because no similar Brownian Bridge procedure is available to perform time-simultaneous statistical inference. In particular, previous implications rely on pointwise considerations.} Note that such an approach would significantly simplify the original justifications given by \cite{lin97} and generalizes his idea mainly used in the context of competing risks \citep{scheike03,hyun2009,beyersmann12b}. In addition, we plan to extend the utilized wild bootstrap technique to general semiparametric regression models; 
see \cite{lin00} for an application in the survival context.
{\color{black}
  Current work investigates to which degree the martingale properties presented in this article may be exploited
  to obtain wild bootstrap consistencies for such functionals of Nelson-Aalen estimates
  or for estimators in semiparametric regression models.
  We are confident that the present approach will lead to reliable inference procedures in these contexts
  for which there has been only little research on such general methodology.
}

In contrast to the procedure of \citet{schoenfeld2002} and other recent publications mentioned in the introduction, the more general illness-death model with recovery does not rely on a constant hazards assumption and captures both the time-dependent structure of mechanical ventilation and the competing event `death in ICU'. This significantly improves medical interpretations. The widths of the confidence bands were competitive compared to the pointwise confidence intervals, i.e., demonstrated usefulness in practical situations. Applications of our theory are not restricted to studies investigating mechanical ventilation, but may also be helpful to investigate, for instance, the impact of immunosuppressive therapy in leukemia diagnosed patients \citep[cf.][]{schmoor2013}. {\color{black} The proposed procedure has even been applied in a recent study investigating femoral fracture risk in an elderly population \citep{bluhmki2016}.}

It has to be emphasized that our simulation study suggested that the 
{\color{black}
wild bootstrap approach leads to more powerful procedures (i.e. to narrower confidence bands) compared to the approximation via Brownian bridges. }
As expected, the applied log-transformation results in improved small sample properties compared to the untransformed {\color{black}wild bootstrap} bands. 
Based on the current simulation study, however, it was difficult to clearly recommend which type of band and which type of multiplier should be used. 

\appendix

\section{Proofs}\label{app}

{\color{black}
\begin{proof}[Proof of Lemma~\ref{lem:mart}]
 Due to similarity, it is enough to concentrate at the first component only; thus, we subsequently suppress the subscript `1'.
 The independence of all white noise processes immediately imply the orthogonality of all component processes.
 At first, we verify the martingale property; the square-integra\-bility is obviously fulfilled since $E(G_1^2(0)) < \infty$.
 To this end, let $0 \leq s \leq t$. 
 By measurability of the counting and the predictable process with respect to $\mathcal{C}_0$, we have
 \begin{align*}
  E(\hat W_n(t) \ | \ \mathcal{C}_s ) = & \sqrt{n} \int_{(0,t]} \frac{J(u)}{Y(u)} E( G(u) \ | \ \mathcal{C}_s) \ d N(u) \\
   = & \sqrt{n} \int_{(0,s]} \frac{J(u)}{Y(u)} G(u) \ d N(u) + \sqrt{n} \int_{(s,t]} \frac{J(u)}{Y(u)} E(G(u)) \ d N(u) = \hat W_n(s)
 \end{align*}
 by the independence of $\sigma(G(u))$ and $\mathcal C_s$ for all $u > s$.
 Hence, the martingale property is shown.
 
 The predictable variation process $\langle \hat W_{n} \rangle$ is the compensator of $\hat W_n^2$, i.e. we calculate
 \begin{align*}
  & E(\hat W_n^2(t) \ | \ \mathcal{C}_s)  \\
  & = n \Big( \int_{(0,s]} \int_{(0,s]} + \int_{(s,t]} \int_{(0,s]} + \int_{(0,s]} \int_{(s,t]} + \int_{(s,t]} \int_{(s,t]} \Big) E(G(u) G(v) \ | \ \mathcal{C}_s ) \\
  & \quad \times \frac{J(u) J(v)}{Y(u) Y(v)}\ d N(u) \ d N(v) \\
  & = n \Big( \int_{(0,s]} \int_{(0,s]} G(u) G(v) + \int_{(s,t]} \int_{(0,s]} E(G(u)) G(v) \\
  & \quad + \int_{(0,s]} \int_{(s,t]} G(u) E(G(v)) + \int_{(s,t]} \int_{(s,t]} E(G(u) G(v)) \Big) \frac{J(u) J(v)}{Y(u) Y(v)} \ d N(u) \ d N(v) \\
  & = n \Big( \int_{(0,s]} \int_{(0,s]} G(u) G(v) \frac{J(u) J(v)}{Y(u) Y(v)} \ dN(u) d N(v) + \int_{(s,t]} E(G^2(u)) \frac{J(u)}{Y^2(u)} \ d N(u) \Big) \\
  & = \hat W_n^2(s) + n \int_{(0,t]} \frac{J(u)}{Y^2(u)} \ d N(u) - n \int_{(0,s]} \frac{J(u)}{Y^2(u)} \ d N(u),
 \end{align*}
 again by the $\mathcal  C_s$-measurability of $G(u)$ for $u \leq s$ and their independence for $u > s$.
 The second to last equality is due to the independence of $G(u)$ and $G(v)$ for $u \neq v$.
 Hence, $(\hat W_n^2(t) - n \int_{(0,t]} \frac{J(u)}{Y^2(u)} \ d N(u))_{t \in [0,\tau]}$ is a martingale.
 
 Letting $\Delta f$ denote the jump-size process fo a c\`adl\`ag function $f$, the definition of the optional variation process yields
 \begin{align*}
  [\hat W_n ] (t)  = \sum_{0 < s \leq t} (\Delta \hat W_n(s) )^2 = n \sum_{0 < s \leq t} G^2(u) \frac{J(u)}{Y^2(u)} \Delta N(u) 
  = n \int_{(0,t]} G^2(u) \frac{J(u)}{Y^2(u)} \ d N(u),
 \end{align*}
 where the sum is taken over all jump points of $N$.
\end{proof}

\begin{proof}[Proof of Theorem~\ref{Th.Na.uni-What}]
 It is enough to verify the conditions of Rebolledo's martingale central limit theorem (in conditional probability); 
 see e.g. Theorem~II.5.1 in \cite{abgk93}.
 Since the filtration $\mathcal C_0$ at time $s=0$ is not trivial, the resulting weak convergence will hold given $\mathcal C_0$ as well, in probability.
 From the classical theory we know that the Aalen-type variance estimator, which is in fact the predictable variation process of $\hat W_n$,
 is uniformly consistent for the variance function.
 
 It remains to prove the Lindeberg condition (2.5.3) on page 83 in \cite{abgk93}.
 But, by the same arguments as in the proof of Lemma~\ref{lem:mart}, this is exactly the same as the Lindeberg condition for the Nelson-Aalen estimator itself.
 And this holds due to the main assumption~\eqref{eq:NA}.
 
 Hence, Rebolledo's martingale central limit theorem yields the desired weak convergence as well as the uniform consistency of the optional variation process.
\end{proof}

}

\begin{proof}[Proof of Theorem~\ref{thm:delta_meth_conv}]
 For convergence~\eqref{eq:weak_conv_logA.1}, see Section~IV.1 in \cite{abgk93} in combination with Slutsky's theorem.
 Convergence~\eqref{eq:weak_conv_logA.2} follows from the consistency of $\sigma^{*2}$, Slutsky's theorem and Theorem~\ref{Th.Na.uni-What}, since
 $\hat W_n$ asymptotically mimicks the distribution of $\sqrt n ( \hat A_n - A)$.
 The functional delta-method for $(x \mapsto \log x)$ completes the proof.
\end{proof}

\begin{proof}[Proof of Corollaries~\ref{cor:CBs} and ~\ref{cor:ks_test}]
 Due to the continuous limit distribution the conditional quantiles converge as well in probability;
 see e.g. \cite{janssen03}, Lemma~1.
 The consistency of $\varphi_n^{KS}$ under $K_{\neq}$ follows from the convergence in probability of the conditional quantile towards a finite value and from the uniform consistency of the multivariate Nelson-Aalen estimator for the cumulative hazard functions.
 Since the factor $\sqrt{n}$ tends to infinity, the test statistic also goes to infinity in probability under $K_{\neq}$.
\end{proof}

\begin{proof}[Proof of Corollary~\ref{cor:equivalence}]
 The proof extends the arguments of \cite{wellek10}, Section 3.1,
 from confidence intervals to confidence bands.
 Write $H = H_1 \cup H_2$ where  
 \begin{align*}
  & H_1 : \{ A(s) \leq A_0(s) - \ell(s) \text{ for some } s \in [t_1,t_2] \} \\
  \text{and} \quad & H_2 : \{ A(s) \geq A_0(s) + u(s) \text{ for some } s \in [t_1,t_2] \}.
 \end{align*}
 Suppose $H$ is true and let without loss of generality be $H_1$ true due to analogy.
 Then the probability of a false rejection of $H$ amounts to
 \begin{align*}
  & P( A_0(s) - \ell(s) < a_n(s) \text{ and } b_n(s) < A_0(s) + u(s) \text{ for all } s \in [t_1,t_2] ) \\
  & \leq P( A_0(s) - \ell(s) < a_n(s) \text{ for all } s \in [t_1,t_2] ) \\
  & \leq P( A(s) < a_n(s) \text{ for some } s \in [t_1,t_2] ) \longrightarrow \alpha.
 \end{align*}
 Here the last inequality holds since $H_1$ is true 
 and the convergence is due to the asymptotic coverage probability of the confidence band $(a_n(s), \infty)_{s \in [t_1,t_2]}$.
 
 In order to prove consistency, suppose the alternative hypothesis $K$ is true and choose any $\varepsilon$ such that 
 $$ 0 < \varepsilon < \inf_{s \in [t_1,t_2]} -(A_0(s) - \ell(s) - A(s)) \wedge (A_0(s) + u(s) - A(s)). $$
 Thus, by the (uniform) consistency of the Nelson-Aalen estimator and the wild bootstrap quantiles,
 the probability of a correct rejection of $H$ equals
 \begin{align*}
  & P( A_0(s) - \ell(s) < a_n(s) \text{ and } b_n(s) < A_0(s) + u(s) \text{ for all } s \in [t_1,t_2] ) \\
  & \geq P(A(s) - \varepsilon < a_n(s) \text{ and } b_n(s) < A(s) + \varepsilon \text{ for all } s \in [t_1,t_2]) \longrightarrow 1
 \end{align*}
 as $n \rightarrow \infty$.
 For the convergence in the previous display, also note that $a_n \xrightarrow{\text{ }P\text{ }} A$ as well as $b_n \xrightarrow{\text{ }P\text{ }} A$ uniformly in $[t_1,t_2]$.
\end{proof}

{\color{black}

\begin{proof}[Proof of Theorem~\ref{thm:prop}]
 
 Let $t_0 > 0 $.
 Denote by $\mathfrak{D}_{>0}[t_0,\tau] \subset \mathfrak{D}[t_0,\tau]$ the cone of positive c\`adl\`ag functions that are bounded away from zero.
 It is easy to see that the functional $\phi: \mathfrak{D}^2_{>0}[t_0,\tau] \rightarrow \mathfrak{D}_{>0}[t_0,\tau], \ (f,g) \mapsto \frac{f}{g}$ is Hadamard-differentiable
 tangentially to the set of pairs of continuous functions $\mathfrak C^2[t_0,\tau]$
 with continuous and linear Hadamard-derivative
 $$ \phi'_{(f,g)}:  \mathfrak C^2[t_0,\tau] \rightarrow \mathfrak C[t_0,\tau], \quad (h_1,h_2) \longmapsto \frac{h_1}{g} - h_2 \frac{f}{g^2}. $$
 A simpler Hadamard-differentiability result holds for $\phi$'s restriction to $\tau$,
 i.e. $\phi|_{\tau}: (0,\infty)^2 \ni (f(\tau), g(\tau)) \mapsto \frac{f(\tau)}{g(\tau)}$
 with continuous, linear Hadamard-derivative 
 $${(\phi|_\tau)}'_{(f,g)}:  \R^2 \rightarrow \R, \quad  (h_1(\tau),h_2(\tau)) \longmapsto \frac{h_1(\tau)}{g(\tau)} - h_2(\tau) \frac{f(\tau)}{g^2(\tau)}.$$
 
 Hence, we apply the functional $\delta$-method and the continuous mapping theorem to
 $$\sqrt{\frac{n_1 n_2}{n}} (\phi(\hat A_{n_2}^{(2)}, \hat A_{n_1}^{(1)}) - \phi(A^{(2)}, A^{(1)})) 
 \quad \text{and} \quad 
 \phi'_{(\hat A_{n_2}^{(2)}, \hat A_{n_1}^{(1)})} \Big( \sqrt{\frac{n_1}{n}} \hat W_{n_2}^{(2)}, \sqrt{\frac{n_2}{n}} \hat W_{n_1}^{(1)} \Big), $$
  respectively,
 verifying their equality in distribution in the limit (conditionally in probability for the latter).
 Proceed similarly with the restricted functional $\phi|_{\tau}$.
 Furthermore, the difference functional of both above functionals retains the Hadamard-differentiability tangentially 
 to the set of pairs of continuous functions.
 Our specific choices of the distance $\rho$ are continuous functionals, hence we are able to apply the continuous mapping theorem again.
 To conclude the proof of the asymptotic behaviour of $\varphi_{n_1,n_2}^{\textnormal{prop}}$ under $H_0^{\textnormal{prop}}$,
 note that the particular weight function solves the problem of dividing by zero at $t_0 = 0$.
 
 For the asymptotic power assertion, let $t_1 \in [0,\tau]$ at which $H_0^{\textnormal{prop}}$ is violated.
 Then
$$ \rho \Big( \ \frac{\hat A^{(2)}_{n_2}}{\hat A^{(1)}_{n_1}} \ , \ \frac{\hat A^{(2)}_{n_2}(\tau)}{\hat A^{(1)}_{n_1}(\tau)} \ \Big)$$
 converges in probability to a positive value, whence $T_{n_1,n_2} \stackrel{p}{\rightarrow} \infty$ follows.
 The conditional quantiles, however, still converge to a finite constant in probability by the above arguments.
 
\end{proof}

}

\section{Supplementary Material: Alternative Proof of Theorem~\ref{Th.Na.uni-What}}

Before proving the conditional convergence in distribution stated in Theorem~\ref{Th.Na.uni-What},
we extend the conditional central limit theorem (CCLT) A.1 given \cite{beyersmann12b} to our context. For that purpose, consider $N(\tau)=\sum_{j=1}^k N_{j}(\tau)$ as the random number of totally observed jumps  in $[0,\tau]$.
Due to the general framework only assuming Aalen's multiplicative intensity model, random sums with a random number $N$ of summands occur and need to be analyzed, since each jump of the counting processes requires its own multiplier $G_{j}(u)$ in the resampling scheme. Thus, we state a more general CCLT as given in \cite{beyersmann12b}, where $\| \cdot \|$ denotes the Euclidean norm on $\R^p$, $p \in \N$.
$\mathfrak L$ again denotes the law.

Throughout, the resampled quantities are modelled via projection on a product probability space
$(\Omega_1 \times \Omega_2, \mathcal{A}_1 \otimes \mathcal{A}_2, P_1 \otimes P_2)$, 
where the white noise processes only depend on the second and the data only on the first coordinate.
\begin{theorem}\label{th:wcext}
 Let $\b{\textit{Z}}_{n;l} : (\Omega_1, \mathcal{A}_1, P_1) \rightarrow (\R^p, \mathcal{B}^p), l = 1,\dots, N,$
be a triangular array of $\R^p$ random variables, $ p\in\N$, where $N: (\Omega_1, \mathcal{A}_1, P_1) \rightarrow (\N_0, \mathcal{P}(\N_0))$ is an integer-valued random variable, 
non-decreasing in $n$, such that $N \stackrel{P}{\rightarrow} \infty$ as $n \rightarrow \infty$. 
Let $G_{n;l}: (\Omega_2, \mathcal{A}_2, P_2)\rightarrow (\R, \mathcal{B}), l \in \N,$ be rowwise i.i.d. random variables 
with $E(G_{n;1}) = 0$ and $var(G_{n;1})=1$. 
Modelled on the product space
$(\Omega_1 \times \Omega_2, \mathcal{A}_1 \otimes \mathcal{A}_2, P_1 \otimes P_2)$,
the arrays $(N, \b Z_{n;l}: l\le N)$ and $(G_{n;l})_{l \in \N}$ are independent. 
Suppose that $\b Z_{n;l}$ fulfills the convergences
\begin{eqnarray}
 \max_{1 \leq l \leq N} \| \b Z_{n;l} \| \stackrel{P}{\longrightarrow} 0 \label{eq:thm_cclt_1} \\
 \sum_{l=1}^N \b Z_{n;l} \b Z_{n;l}' \stackrel{P}{\longrightarrow} \boldsymbol{\Gamma}, \label{eq:thm_cclt_2}
\end{eqnarray}
where $\bs \Gamma$ is a positive definite covariance matrix.
Then, conditionally given $(N, \b Z_{n;l}: l\leq N)$, the following weak convergence holds in probability:
\begin{eqnarray}
 \mathcal{L} \Big( \sum_{l=1}^N G_{n;l} \b Z_{n;l} \ \Big| \ N, \b Z_{n;l}: l\leq N \Big) \stackrel{d}{\longrightarrow} N(\bs{0}, \bs \Gamma) \label{eq:thm_cclt_3}.
\end{eqnarray}
\end{theorem}
\begin{proof}
Since $N$ is non-decreasing in $n$ with $N \stackrel{P}{\rightarrow} \infty$, it follows that $N(\omega_1, \omega_2) \rightarrow \infty$ 
for $P_1$-almost all $\omega_1\in\Omega_1$, independently of the value $\omega_2 \in \Omega_2$. 
Thus, for $P_1$-almost all such fixed $\omega_1\in\Omega_1$, 
we have a deterministic number of summands $N(\omega_1,\cdot)$. 
By the subsequence principle, choose a subsequence $(n') \subseteq (n) = \mathbb N$ along which \eqref{eq:thm_cclt_1} and \eqref{eq:thm_cclt_2} hold for almost every $\omega_1 \in \Omega_1$ as well.
Applying the CCLT A.1 in \cite{beyersmann12b} with its conditions being almost surely fulfilled, 
the weak convergence \eqref{eq:thm_cclt_3} follows $P_1$-almost surely along $n'$. 
A further application of the subsequence principle, going back to convergence in probability, completes the proof.
\end{proof}

\begin{proof}[Proof of Theorem~\ref{Th.Na.uni-What}]
{\it Conditional Finite-Dimensional Convergence of $\hat{\b W}_n$}.\\
Due to asymptotic mutual independence, we only consider the first entry $\hat W_{1n}$ of $\hat{\textbf{\textit{W}}}_{n}$
and suppress the subscript `1' subsequently.
Define countably many i.i.d. random variables $ \tilde G_{n;1}, \tilde G_{n;2},\ldots$ with $E(\tilde G_{n;1})=0$ and $var(\tilde G_{n;1})=1$,
that are independent of $\mathcal C_n$,
and define processes $ Z_{n;1},\ldots,Z_{n;N(\tau)}$ such that equation (\ref{NA.uni.Wnh}) is re-expressed as
\begin{align}
\hat{W}_n(t) =\sum_{v\in T}G(v) \sqrt{n} X_v(t) \stackrel{d}{=} \sum\limits_{l=1}^{N(\tau)} \tilde G_{n;l}Z_{n;l}(t), \label{eq:triarray}
\end{align}
where $\stackrel{d}{=}$ denotes equality in distribution.
Here $X_v(t):= \boldsymbol{1} \{ v \leq t \} \Delta N(v) / Y(v) $
and $T=\lbrace u\in[0,\tau] \ \vert \ \Delta N(u)=1\rbrace$ contains all jump times of the counting process $N$.
Then, the general framework of Theorem \ref{th:wcext} is fulfilled for the triangular array $Z_{n;l}(t_j),$ $l=1,\ldots,N(\tau),$ \ $j=1,\dots, r$,
for any finite subset $\{t_1,\dots,t_r\} \subset [0,\tau]$.
Next, conditions~\eqref{eq:thm_cclt_1} and~\eqref{eq:thm_cclt_2} are verified in a similar manner as in \cite{beyersmann12b}.
Applying the subsequence principle for convergence in probability to assumption (\ref{Pre.ass11}), it follows that for every subsequence there exists a further subsequence, say $n$, such that as $n\rightarrow\infty$ 
\begin{align}
\label{eq:YbynAS}
\underset{u\in[0,\tau]}{\sup}\Big| \frac{Y(u)}{n}-y(u)\Big|\xrightarrow{\text{ }a.s.\text{ }}0,
\end{align}
i.e., the left-hand side converges to zero for $P_1$-almost all $\omega\in\Omega_1$.
Fix an arbitrarily small $\epsilon>0$ and an $\omega$ for which \eqref{eq:YbynAS} holds.
The following arguments implicitly consider all $n \geq n_0(\omega,\epsilon)$ for an $n_0$ determined by~\eqref{eq:YbynAS}.
Hence, the left-hand side of~\eqref{eq:YbynAS} is less than $\epsilon$ for all such $n \geq n_0$.
Choose a $\gamma_\epsilon = \gamma_{\epsilon}(\omega) >0$  such that
\begin{align*}
 \underset{{u\in[0,\tau]}}\sup \frac{n}{Y(\omega,u)} \leq \frac{\gamma_{\epsilon}}{y(u)} \leq \frac{\gamma_{\epsilon}}{\underset{{v\in[0,\tau]}}\inf y(v)}=:c_{\epsilon}.
\end{align*}
Since $X_v$ is (at most) a one-jump process on $[0,\tau]$, we have
\begin{align*}
\underset{{l = 1, \dots, N(\tau)}}\sup \ \underset{{t\in[0,\tau]}}\sup | Z_{n;l}(t) |
& \leq \sqrt{n}  \ \underset{{v \in T}}\sup X_{v}(\omega,\tau)
\leq n^{-1/2} \frac{n}{Y(\omega, \tau)} \leq n^{-1/2} {c_\epsilon}\xrightarrow{n\rightarrow \infty} 0 .
\end{align*}
In particular, $\textbf{\textit Z}_{n;l}=( Z_{n;l}(t_1),\ldots,Z_{n;l}(t_k))'$ satisfies 
$ \underset{1 \leq l \leq N(\tau)}\max \| \textbf{\textit Z}_{n;l}(t) \| \stackrel{P}{\longrightarrow} 0 $,
and \eqref{eq:thm_cclt_1} holds.

For simplicity, condition~\eqref{eq:thm_cclt_2} is only shown for two time points $0\le t_1\le t_2\le \tau$,
such that $\textbf{\textit{Z}}_{n;l}=(Z_{n;l}(t_1),Z_{n;l}(t_2))'$. 
Representation (\ref{eq:triarray}) implies that 
\begin{align*}
\sum\limits_{l=1}^{N(\tau)} \textbf{\textit Z}_{n;l} \textbf{\textit Z}_{n;l}'=n \sum_{v\in T} \begin{pmatrix}
 X_v^2(t_1) &X_v(t_1) X_v(t_2)\\
 X_v(t_1) X_v(t_2) & X_v^2(t_2)
\end{pmatrix}.
\end{align*}
The off-diagonals equal $X_v(t_1) X_v(t_2)= \boldsymbol{1} \{ v \leq t_1 \} \Delta N(v) / Y^2(v)$
and the other two components are obtained for $t_1 = t_2$.
Using the Doob-Meyer decomposition~\eqref{eq:doobmeyer}, it follows that
\begin{align*}
n \sum_{v\in T}X_v(t_1) X_v(t_2)
= n^{-1} \underset{(0,t_1]}{\int{}} \Big( \frac{n}{Y(u)} \Big)^2 d M (u)
  + \underset{(0,t_1]}{\int{}} \frac{n}{Y(u)} \alpha(u) d u .
\end{align*}
As in \cite{beyersmann12b}, Rebolledo's martingale central limit theorem (\citealp{abgk93}, Theorem~II.5.1)
shows the negligibility of the martingale integral.
The remaining integral converges to $\psi(t_1,t_2) = \int_{(0,t_1]} \alpha(u) / y(u) d u$ in probability due to assumption~\eqref{eq:YbynAS}.
Consequently, we conclude that, as $n\rightarrow \infty$,
\begin{align*}
\sum\limits_{l=1}^{N(\tau)} \textbf{\textit Z}_{n;l}\textbf{\textit Z}_{n;l}'\xrightarrow{\text{ }P\text{ }}
\begin{pmatrix}
 \psi(t_1,t_1) & \psi(t_1,t_2) \\ \psi(t_1,t_2) & \psi(t_2,t_2)
\end{pmatrix}.
\end{align*}

Let $U$ be a zero-mean Gaussian process with covariance function $\psi$. Extending previous arguments to $r\in\mathbb N$ time points $t_1,\ldots,t_r$, 
Theorem \ref{th:wcext} implies, conditionally on $\mathcal{C}_n$, the finite-dimensional weak convergence 
\begin{align*}
(\hat{W}_n(t_1), \dots, \hat{W}_n(t_r))' \stackrel{d}{\longrightarrow} (U(t_1), \dots, U(t_r))'
\end{align*}
in probability.
Conditionally on $\mathcal{C}_n$, only the white noise processes $G_{1}, \dots, G_k$ in \eqref{NA.uni.Wnh} are random and, in particular, stochastically independent. 
This implies the multivariate conditional weak convergence
\begin{align*}
 (\hat{\textbf{\textit W}}_n(t_1), \dots, \hat{\textbf{\textit W}}_n(t_r))' \stackrel{d}{\longrightarrow} (\textbf{\textit U}(t_1), \dots, \textbf{\textit U}(t_r))'
\quad \text{in probability},
\end{align*}
where $\textbf{\textit U} = (U_1, \dots, U_k)'$ has independent components and the asserted covariance structure.
\vspace{0.5cm}

The {\it conditional tightness of $\hat{\b W}_n$}
 follows similarly as in the proof of Theorem~3.1 in \cite{dobler14}.
 As previously, tightness of $\hat{\b{\textit{W}}}_n$ is separately studied for each single component, i.e., we only consider $\hat W_{jn}$ 
 and suppress the subscript `$j$' of the estimators and counting processes as above.
 Let $0\leq r \leq s \leq t \leq \tau$. Then, Theorem~15.6 in \cite{billingsley68} using $\gamma = 2$ and $\alpha = 1$ in combination with the remark on p. 356 in \cite{jacod03} leads us to the following conditional expectation: 
 \begin{align*}
  & E[(\hat{W}_n(t) - \hat{W}_n(s))^2(\hat{W}_n(s) - \hat{W}_n(r))^2 \ \left| \right. \ {\color{purple}\mathcal C_0} ] \\
  & = n^2 E\Big[ \Big( \int\limits_{(s,t]} G(u) \frac{J(u)}{Y(u)} d N(u) \Big)^2 
    \Big( \int\limits_{(r,s]} G(v) \frac{J(v)}{Y(v)} d N(v) \Big)^2 \ \Big| \ {\color{purple}\mathcal C_0} \Big]\\
  & = n^2  \int\limits_{(s,t]} \int\limits_{(s,t]} \int\limits_{(r,s]} \int\limits_{(r,s]}
    \frac{J(u_1)}{Y(u_1)}  \frac{J(u_2)}{Y(u_2)} \frac{J(v_1)}{Y(v_1)} \frac{J(v_2)}{Y(v_2)} \\
    & \quad \times E[ G(u_1) G(u_2) G(v_1) G(v_2) ]
    d N(v_2) d N(v_1) d N(u_2) d N(u_1).
 \end{align*}
  Since the multipliers $G(u), \ u \in T $ 
 are independent and the intervals $(r,s]$ and $(s,t]$ are disjoint, 
 the remaining expectation decomposes into a product of  $E[G(u)]$ or $E[G^2(u)]$.
 Here, each expectation of a multiplier to the power of one vanishes due to $E[G(u)] = 0$
 and a multiplier to the power of two only occurs
 whenever $u_1 = u_2 \in T$ or $v_1 = v_2 \in T$.
 Since $E[G^2(u)] = 1$, the above display simplifies to
 \begin{align*}
   n^2 \int\limits_{(s,t]} \frac{J(u)}{Y(u)^2} d N(u)
    \int\limits_{(r,s]} \frac{J(v)}{Y(v)^2} d N(v) 
   = [ \hat\sigma^2(t) - \hat\sigma^2(s) ] [ \hat\sigma^2(s) - \hat\sigma^2(r) ]
    \leq [ \hat\sigma^2(t) - \hat\sigma^2(r)]^2
 \end{align*}
 with $\hat \sigma^2$ defined as in \eqref{eq:varaalen}.
 By Theorem~IV.1.2 in \cite{abgk93} the convergence in probability of the right-hand side to $ (\sigma^2(t) - \sigma^2(r))^2 $ holds uniformly in $r,t \in [0,\tau]$.
 Following the lines of \cite{dobler14} by utilizing the proposition in \cite{jacod03}, p. 356, conditional tightness is shown along subsubsequences almost surely.
 Another application of the subsequence principle shows the stated result.
\end{proof}

\section*{Acknowledgements}
Jan Beyersmann was supported by Grant BE 4500/1-1 of the German Research Foundation (DFG).

\bibliographystyle{plainnat}
\bibliography{literatur}

\end{document}